%% file: main.tex
\documentclass{article}
\input{macros}
\title{Interactive coding resilient to an unknown number of erasures}

\author{
	Ran Gelles\thanks{Faculty of Engineering, Bar-Ilan University. \texttt{ran.gelles@biu.ac.il}} 
\and 
	Siddharth V.\@ Iyer\thanks{University of Washington. \texttt{siyer@cs.washington.edu}. Part of this work done while at Bar-Ilan University and a student at Birla Institute of Technology and Science, Pilani.}
}

\date{}

\begin{document}
\maketitle

\begin{abstract}
We consider distributed computations between two parties carried out over a noisy channel
that may erase messages. 
Following a noise model proposed by Dani et al. (2018), the noise level observed by the parties 
during the computation in our setting is arbitrary and a priori unknown to the parties.

We develop 
\emph{interactive coding schemes} that adapt to the actual level of noise
and correctly execute any two-party computation. Namely, in case the channel erases~$T$ 
transmissions, the coding scheme will take $N+2T$ transmissions using an alphabet of size~$4$ (alternatively, using $2N+4T$ transmissions over a binary channel) 
to correctly simulate any binary protocol that takes $N$ transmissions assuming a 
noiseless channel. 
We can further reduce the communication to~$N+T$ 
by relaxing the communication model 
and allowing parties to remain silent rather than forcing them to communicate 
in every round of the coding scheme.

Our coding schemes are efficient, deterministic, have linear overhead both in their communication and round complexity, and succeed (with probability~1) regardless of the number of erasures~$T$.
\end{abstract}

\input{intro}

\input{model}
\input{sim}
\input{correctness}
\input{sim3}

\input{corr3}

\section*{Acknowledgements}
We thank Amir Leshem for plenty helpful discussions.
Research supported in part by the Israel Science Foundation (ISF) through grant No.\@ 1078/17.

\bibliographystyle{alphabbrv-doi}
\bibliography{coding}

\newpage
\appendix
\section*{Appendix}

\input{appendix}

\end{document}

%% file: macros.tex
\usepackage{textcomp, amsmath, amsthm, amssymb,mathtools,bm}
\usepackage[utf8]{inputenc}
\usepackage{thmtools}
\usepackage{microtype}
\usepackage[ruled,hangingcomment]{algorithm2e}
\usepackage{algpseudocode}
\usepackage{nameref}
\usepackage{float}
\usepackage[outline]{contour}
\usepackage[table,svgnames]{xcolor}
\usepackage{breakurl}
\usepackage[hyperindex,breaklinks,colorlinks=true,linkcolor=black,citecolor=MidnightBlue,urlcolor=MidnightBlue]{hyperref}
\usepackage{hyperref}
\usepackage{fullpage}
\usepackage{xspace}

\newtheorem{theorem}{Theorem}[section]
\newtheorem{lemma}[theorem]{Lemma}
\newtheorem{claim}[theorem]{Claim}

\newtheorem{corollary}[theorem]{Corollary}
\newtheorem{definition}[theorem]{Definition}

\newcommand\numberthis{\addtocounter{equation}{1}\tag{\theequation}}

\newcommand{\N}{\mathbb{N}}

\providecommand{\calT}{{\mathcal T}{}}

\newcommand{\Inf}{\mbox{{\sf Info}}}
\newcommand{\Par}{\mbox{{\sf Parity}}}
\newcommand{\sr}{round\xspace}
\newcommand{\srs}{rounds\xspace}

\newcommand{\emptyword}{\lambda}

\newcommand{\CC}{\mathsf{CC}}
\newcommand{\ca}{\mathsf{c_A}}
\newcommand{\cb}{\mathsf{c_B}}
\newcommand{\ch}{\mathsf{c_{ch}}}
\newcommand{\RC}{\mathsf{RC}}

\newcommand{\erasure}{\bot}
\newcommand{\silence}{\square}

%% file: intro.tex
\section{Introduction} 

Consider two remote parties that use a communication channel in order
to perform some distributed computation. One main obstacle they may face is 
noise added by the communication channel, corrupting their messages and ruining the computation. In the early 90's, Schulman~\cite{schulman92,schulman96}
initiated the field of \emph{interactive coding} where the parties use
coding techniques in order to complete their computation correctly despite
possible communication noise.\looseness=-1

Channels may introduce different types of noise.
Among the common noise types are \emph{substitution} noise,
where the channel changes the content of messages (e.g., it flips communicated bits),
\emph{insertions and deletions}, where the channel may introduce a new message or completely remove
a transmission, and \emph{erasures} where the channel erases transmissions, i.e., replacing them with
an erasure mark.
Throughout the last several years many interactive coding schemes
were developed, allowing parties to perform computations over the various channels and noise types,
e.g., \cite{BR14,EGH16,SW17}, see related work below and \cite{gelles17} for a survey.

Naturally, some bounds on the noise must be given. As a trivial example, it is clear that if \emph{all} the transmissions are  corrupted, there is no hope to complete any distributed computation. 
Braverman and Rao~\cite{BR14} showed that for substitution noise, a noise fraction of $1/4$ is maximal. That is, as long as the noise is limited to corrupt less than 
$1/4$ of the communication, there are coding schemes that will succeed in producing the correct output. 
However, if the noise exceeds this level,
\emph{any} coding scheme for a general computation is bound to fail.
Similarly, noise rate of $1/4$ is maximal (and achievable) for insertion and deletion noise~\cite{BGMO17,SW17}, and
noise rate of $1/2$ is maximal  (and achievable) for erasure noise~\cite{FGOS15,EGH16}.

All the above schemes must know \emph{in advance} the amount of noise they are required to withstand.
For instance, the scheme of~\cite{BR14} is given some parameter~$\rho<1/4$, which stands for the  (maximal) fraction of noise that instance will be able to handle.
Given this parameter and the length of the (noiseless) computation to be performed~$N$, 
the coding scheme determines how many transmissions it should 
take in order to perform the computation, say, $\tilde N = \tilde N(\rho,N)$ transmissions.
It is then guaranteed that as long as the noise corrupts at most $\rho\tilde N$ transmissions, the computation succeeds.

In a recent work, Dani et al.~\cite{DHMSY18} considered the case where the noise amount may be arbitrary and \emph{a priori} unknown to the coding scheme. That is, the channel may corrupt up to~$T$ transmissions, where $T\in\mathbb{N}$ is some fixed amount of noise which is independent of the other parameters of the scheme and the computation to be performed. 
The work of Dani et al.~\cite{DHMSY18} considered substitution noise and showed a scheme that succeeds
with high probability and, if the channel corrupts~$T$ transmissions during the execution of the scheme, 
then the scheme will take $N+O(T+\sqrt{T N \log N})$ transmissions to conclude.

Again, some limitations must be placed on the noise. 
Indeed, assume that an uncorrupted computation on the input $(x,y)$ terminates in~$\tilde N$ rounds, and assume~$T>\tilde N$. Then, a substitution noise can always make Alice (wlog) believe that Bob holds $y$, by corrupting the messages Alice receives in the first $\tilde N$ rounds. Note that Alice then terminates with the wrong output, i.e., the output for~$(x,y)$.
Dani et al.~\cite{DHMSY18} dealt with the above impossibility  by assuming that the parties 
have access to some shared randomness and that the adversary is \emph{oblivious} to that randomness. This allowed them to employ cryptographic tools in order to guarantee the authentication of communicated messages.

\medskip

In this work we focus on \emph{erasure noise}, where the channel either transmits the message as is, 
or  erases the message and outputs a special~$\bot$ symbol that indicates this event. 
Erasure noise naturally appears in many practical situations, e.g., when an Ethernet packet gets corrupted yet this corruption is detected by the CRC checksum mechanism~\cite{802.3}. In this case the packet is considered invalid and will be dropped. This situation is equivalent to an erasure of that transmission.
We note that this type of noise is weaker than substitution noise. This allows us to obtain coding schemes without further assumptions, mainly, 
without the need for shared randomness and without requiring the adversary to be oblivious.

Our main result is an efficient and deterministic coding scheme
that withstands 
an arbitrary and \emph{a priori} unknown $T$ erasures with probability~$1$.
\begin{theorem}[main, informal]\label{thm:main}
Given any two-party binary interactive protocol $\pi$ of length~$N$,
there exists an efficient and deterministic noise-resilient protocol~$\Pi_4$ of length~$N+2T$ 
that simulates~$\pi$ over a 4-ary channel in the presence of an arbitrary and a priori unknown number~$T$ of erasures. 

Alternatively, there exists an efficient, deterministic, \emph{binary} noise-resilient protocol~$\Pi_2$ of length~$2N+4T$ that simulates~$\pi$ assuming an arbitrary and a priori unknown number~$T$ of erasures.
It holds that
\[
\CC(\Pi_4)=\CC(\Pi_2) = 2N+4T.
\]
\end{theorem}

Since~$T$, the amount of noise, is unknown to begin with, the length of the coding scheme must 
adapt to the actual noise that the parties observe. 
Such coding schemes are called \emph{adaptive}~\cite{AGS16}.
Adaptivity raises several issues that must be dealt with appropriately.
The main issue is \emph{termination}. Since the coding scheme adapts its length to the observed noise,
and since the different parties observe different noise patterns, their termination is
not necessarily synchronized. 
As a matter of fact, obtaining synchronized termination is impossible. This can be seen as a variant of the famous ``coordinated attack problem''~\cite{HM90}, where reaching full synchronization between two parties is known to be impossible.  See Appendix~\ref{app:modeldetails} for an elaborated discussion about unsynchronized termination.

Unsynchronized termination means that one party may terminate while the other party 
continues to send (and receive) transmissions as dictated by its protocol.
In this case, the communication model should 
specify what happens in those rounds where only one party is active and the other has terminated. 
Specifically, it should specify what messages the active party receives in this case.

In our setting we define a special symbol we call \emph{silence} (cf.~\cite{AGS16}). 
We assume that silence is (implicitly) communicated by a terminated party. That is, the still-active party hears silence in every round it is set to listen and the other party has already terminated.
We note that silence is corruptible---the channel may erase silent transmissions, and the active party will see an erasure mark instead. On the other hand, these implicit silent transmissions are not considered part of the communication of the protocol (i.e., we do not count them towards the communication complexity). 
Hearing a silence is a univocal indication that the other side has terminated, allowing the other party to terminate as well and bypassing the impossibility of synchronized termination (Lemma~\ref{lem:termination}).

\medskip

In the setting of Theorem~\ref{thm:main} we do not allow the parties to remain silent before termination---in each timestep a party is set to speak it must communicate a valid message. However, in today's networks, especially in networks with multiple parties, it is very common that parties send 
messages only if they have information to send, and keep silent otherwise. 

Our second result extends Theorem~\ref{thm:main} to the setting where parties are allowed 
to either speak or remain silent at every timestep (called the AGS setting hereinafter, see~\cite{AGS16}). In this case, termination becomes even more tricky. Recall that the coding scheme of Theorem~\ref{thm:main} uses 
silence as an indicator for termination.
We can do the same in the AGS setting and avoid using silence throughout the protocol, keeping it as an indicator towards termination only. However, this would effectively reduce the AGS setting to the one of Theorem~\ref{thm:main}, and lead to a suboptimal scheme.

Instead, we take a different approach that requires parties to remain silent during the protocol in certain cases.
This has the effect of reducing the communication at the expense of 
not being able to identify termination at times. 
In particular, one of the parties may remain active indefinitely.
However, that party will remain silent after the other party has terminated, and moreover, it will hold the correct output of the computation. We call this situation \emph{semi-termination} (see also Definition~\ref{def:semi-term}).
Our result for this setting is as follows. 
\begin{theorem}[AGS setting, informal]\label{thm:main2}
Given any binary two-party interactive protocol~$\pi$ of length~$N$,
there exists an efficient and deterministic protocol~$\Pi_4$ in the AGS setting with semi-termination, of length~$N+4T$ sending at most $N+T$ symbols from alphabet of size~$4$,  
that simulates~$\pi$ assuming an arbitrary and a priori unknown number~$T$ of erasures.

Alternatively, There exists an efficient and deterministic protocol~$\Pi_1$ in the AGS setting with semi-termination,
with length $4N+16N$ that communicates at most~$N+T$ unary (i.e., non-silent) symbols and
simulates $\pi$ assuming
an arbitrary and a priori unknown number~$T$ of erasures.
\end{theorem}
In this setting, it is significant to bound the round complexity 
since silence is not counted towards the communication complexity, yet it can be used to transfer information. 
As stated in the above theorem, the round complexity of the resilient protocols~$\Pi_4$ and~$\Pi_1$ until the event of semi-termination, is linear in the round complexity of the noiseless protocol~$\pi$, namely, $\RC(\Pi_4)=O(N+T)$ and $\RC(\Pi_1) =O(N+T)$.

\subparagraph{Organization.}
In the next part we overview the techniques used to obtain Theorem~\ref{thm:main} and Theorem~\ref{thm:main2}. Section~\ref{sec:model} formally defines the communication and noise model, and fixes some notations used throughout.
In Section~\ref{sec:sim} we describe the noise resilient protocol of Theorem~\ref{thm:main} and analyze its correctness. 
The coding scheme of Theorem~\ref{thm:main2} along with its correctness analysis appear in Section~\ref{sec:adaptive}. 
Finally, in Appendix~\ref{app:modeldetails} we discuss the issue of unsynchronized termination and prove that it is impossible to obtain synchronized termination when the noise is unbounded.  

\subsection{Coding Schemes Overview}
Erasure noise has two attractive properties we utilize towards our scheme.
The first is that, if there was a corruption, the \emph{receiver} is aware of this event; the second property is
that, if there was no corruption, the received message is the \emph{correct} one.
Our scheme follows a technique by Efremenko et al.~\cite{EGH16}, where the parties
simulate the noiseless protocol~$\pi$ bit by bit. As long as there is no noise, they can carry
out the computation identically to~$\pi$. However, in the case of an erasure,
the receiver needs to signal the other side that it did not receive the last message
and request a retransmission. The main problem is that this request message may get
erased as well, making both sides confused regarding to what should be sent next.

In~\cite{EGH16} this issue is solved by extending each message by two bits that indicate
the round currently being simulated. 
It is proven that the parties may simulate different rounds, however, the
discrepancy in the round number is bounded by~$1$. Hence, a round number modulus~$3$
is required in~\cite{EGH16} to indicate whether Alice is ahead, Bob is ahead, or they are at the 
same round.

Our scheme combines the above technique with a challenge-response technique employed 
by Dani et al.\@ in~\cite{DHMSY18}, in order to obtain resilience against an unbounded number of erasures.
 Our coding scheme is \emph{not} symmetric, but rather
Alice always begins a round by sending a ``challenge'' message, followed by Bob replying with a response.
Alice then determines whether the challenge-response round was successful:
if both messages were not erased, Alice would see the correct response from Bob and would deduce
that both messages were received correctly. If Alice received an erasure, or if she received
the wrong response, she would deduce that an erasure has occurred during this round and the round should be re-simulated.

Bob, similarly but not identically,  gets the challenge message from Alice and verifies that it
belongs to a new round (that is, the challenge differs from the previous round).
If this is the case, he replies with the next bit. Otherwise, i.e., if Alice's challenge was erased, or if 
Bob receives the challenge of the previous round, he replies with the response of the \emph{previous} round.
Note that if Alice did send a new challenge and it was erased, she will now get the response of the previous
round and realize there is a mismatch.

It is not too difficult to see that the ``challenge'' suffices to be a single bit---the current simulated round number in~$\pi$, modulus~$2$ (which we call the \emph{parity} from now on). 
Since the scheme is not symmetric, it can never happen that Alice has advanced to the next round
while Bob has not. The other case, where Bob is ahead of Alice by a single round is still possible.
Therefore, one bit of information suffices to distinguish between these two cases.

In more details, Alice begins a round by sending Bob the next bit of the simulation of~$\pi$, along with 
the parity of the round number (in~$\pi$) of that bit. If Bob receives this message and the parity
matches the round number he expects, he records the bit sent by Alice and replies with the next bit of~$\pi$ using the same parity. If this message reaches Alice correctly, she knows the round is over and advances to the next round.
If Bob does not receive Alice's message (i.e., it gets erased), or if the parity is incorrect, Bob replies with the bit of the previous round along with the parity of that round. Similarly, if Alice receives a message with a wrong parity or an erased message, 
she keeps re-simulating the same round, until she gets the proper reply from Bob.

Note that a single erasure delays the progress of the simulation by a single round (2~transmissions). 
However, once there is a round in which both messages are not erased, the simulation correctly
continues, and the succeeding two bits of~$\pi$ are correctly simulated. That is, as long as there is noise,
the simulation just hangs, and when the noise ceases, the simulation continues from exactly the same
place it stopped.

Once Alice completed simulating the last round of~$\pi$, she quits. Recall that Bob is always ahead of Alice,
thus if Alice completes the simulation, so does Bob. After Alice terminates, Bob receives silence unless erased by
the channel. When Bob hears a silence, he learns that Alice has terminated and quits the scheme as well.
The noise may delay Bob's time of termination by corrupting  the silence, 
however, once the noise is over, Bob will learn that the simulation has completed and will quit as well.

\paragraph*{Coding schemes for parties who may keep silent.}

Our second scheme works in the communication model where parties are allowed to remain silent if they wish (the AGS setting).
The main advantage in being able to remain silent is allowing the parties to communicate information in an optimized manner 
which reduces their communication complexity. 
Specifically, consider the above idea of challenge-response, 
where a party replies with the wrong response
in order to indicate there was an error and that a round should be re-simulated.
In the AGS setting we can instead keep silent in order to signal this retransmission request.

The idea is as follows. Similar to the scheme above Alice begins a round by sending
her bit to Bob (along with the parity). If the transmission is received correctly, Bob replies with his next bit.
If Bob receives an erasure instead, he remains silent. This signals Alice
that an error has happened and that she should re-simulate the last round.
Similarly, if Alice sees an erasure she keeps silent.
In all other cases, i.e., receiving a silence or receiving a wrong parity, 
 the parties re-transmit their last message as before.

The effect of a party keeping silent for asking a retransmission is reducing the communication
complexity. Note that each erasure causes the recipient to remain silent for one single round,
instead of sending a message that indicates an erasure. 
Then, the simulation continues from the point it stopped. 
On the other hand, the analysis becomes slightly more difficult
in this case since silence may be erased as well, causing the other side to remain silent
and signal there was an erasure. 
This may cause the first party to repeat its message
while the other side should have actually resent its message in order to advance the simulation.
Luckily, this issue does not falsify the correctness of the simulation---as a result of sending the parity, the extra retransmission is simply ignored. 
Furthermore, such a superfluous transmission does not increase the communication overhead 
since it can only happen when multiple erasures have occurred in the same round or in consecutive rounds.

When silence has a meaning of requesting a retransmission, we cannot use it anymore to indicate termination.
Note that whatever Alice sends Bob to inform him she is going to terminate may get erased,
so if Alice terminates Bob will not be aware of this fact and will remain active. If Alice waits to hear a confirmation from Bob
that he received the indication and learned that Alice is about to terminate---this confirmation may get erased.
Bob never learns if Alice has received his acknowledgment or not; then, if Bob assumes that 
Alice has terminated and terminates himself, 
it will be Alice who hangs in the protocol waiting for Bob's confirmation, and so on.
Our approach to this conundrum is to allow the parties not to terminate as long as
there exists a point in time beyond which the parties remain silent and both hold the correct output.
In our scheme, Alice will actually terminate once she learns that the simulation is done%
\footnote{It is possible to let Alice send Bob a ``termination message'' 
right before she quits in order to signal Bob she is about to terminate.
In case this message is not erased, Bob will terminate as soon as he hears this special message. 
Otherwise, he will keep running the scheme but remain silent. 
In either case, the protocol satisfies the semi-termination requirements.}
(recall that if Alice completed the simulation then Bob completed it as well, but the other direction does not
necessarily hold). 
Bob's actions in the final part of the protocol are slightly different from his normal behaviour. 
Once Bob completed simulating~$\pi$, but he is unaware whether or not Alice completed simulating~$\pi$,
he keeps silent unless he hears a message from Alice that re-simulates the final round.
In this case, he replies with his final bit. In all other cases he remains silent.

\subsection{Related Work}
The field of interactive coding was initiated by the seminal work of Schulman~\cite{schulman92,schulman96} focusing on two parties that
communicate over a binary channel with either substitution noise (random 
or adversarial) or with erasure noise. Followup work (for substitution noise)
developed coding schemes with optimal resilience~\cite{BR14,BE17,GH14}
efficiency~\cite{schulman03,braverman12,GMS14,BKN14,GH14,GHKRW18},
or good rate~\cite{KR13,haeupler14,GHKRW18}.  
Coding schemes for different channels and noise types were developed in~\cite{pankratov13,EGH16,GH17} for channels with feedback,
in~\cite{BGMO17,SW17,HSV18,EHK18} for insertions and deletions noise,
and in~\cite{BNTTU14,LNSTYY18} for quantum channels.
Interactive coding over channels that introduce erasure noise was explored in~\cite{schulman96,pankratov13,FGOS15,GH17,EGH16,AGS16}. In particular, Efremenko et al.~\cite{EGH16} developed
efficient coding schemes for optimal erasure rates, and Gelles and Haeupler~\cite{GH17} developed efficient coding schemes with optimal rate assuming a small fraction of erasures. Adaptive models were considered in~\cite{AGS16,GHS14,GH14}.
See~\cite{gelles17} for a survey on the field of interactive coding.

\smallskip
Closest to our work is the work of Dani et al.~\cite{DHMSY18} who considered the case of an arbitrary noise amount
that is unknown to the scheme. Their coding scheme assumes substitution noise which is oblivious to the randomness used by the parties as well as to the bits communicated through the channel. An AMD code~\cite{CDFPWS08} is used to fingerprint each transmitted message, allowing the other side to detect corruptions with high probability.
Aggrawal et al.~\cite{ADHS18} use similar techniques to develop a robust
protocol for message transfer, assuming bi-directional channel that suffers from an arbitrary (yet finite) and unknown amount of bit flips. Aggarwal et al.~\cite{ADHS17} extended this setting to the multiparty case, 
where $n$~parties, rather than two, perform a computation over a noisy network with an arbitrary and unknown amount of noise.

%% file: model.tex
\section{Model Definition}\label{sec:model}
\subparagraph*{Standard notations.}
For an integer $i\in \N$ we denote by $[i]$ the set $\{1,2,\ldots, i\}$. All logarithms are taken to base~2.
The concatenation of two strings $x$ and $y$ is denoted $x\circ y$. We let $\emptyword$ denote the empty string.

\subparagraph*{Interactive computation.}
In our setting, two parties, Alice and Bob, possess private inputs $x\in \{0,1\}^n$ and $y\in \{0,1\}^n$, respectively, and
wish to compute some predefined function $f(x,y)$.
The computation is performed by exchanging messages over a channel with fixed alphabet~$\Sigma$.
The computation is specified by a synchronous \emph{interactive protocol}~$\pi$.
An interactive protocol $\pi=(\pi_A,\pi_B)$ is a pair of algorithms 
that share a common clock, and specify, for each timestep and each one of the parties, the following details: 
(1) which party speaks and which party listens in this time-step; 
(2) if the party is set to speak, which symbol to communicate;
(3) whether the party terminates in this timestep, and if so, what the output is.

Without loss of generality, we assume that the (noiseless) protocol~$\pi$ is alternating, i.e., Alice and Bob speak in an alternating manner,
Alice speaks in odd rounds and Bob in even rounds. 
If this is not the case, it can be made so while increasing the communication by a factor of at most~2.
We define a \textit{round} of the noiseless protocol to be a sequence of two time steps in which two consecutive messages are sent: the first is sent by Alice and the second by Bob. For example, the first round consists of the first message sent by Alice and the subsequent message from Bob. In general for $r\geq 1$, after $r-1$ rounds have elapsed, the $r^{th}$ round consists of the $r^{th}$ message sent by Alice and Bob's subsequent message. For the sake of convenience, we assume that the last message of the protocol is sent by Bob; this can be ensured by padding the protocol by at most one bit. Since the protocol is alternating we can also think of rounds being associated with timesteps. More formally, in the $r^{th}$ round, Alice and Bob send messages corresponding to the $(2r-1)^{th}$ and the $2r^{th}$ timestep respectively. We say that a protocol has length $N$ if Alice and Bob exchange $N$ messages in the protocol.

Given a specific input~$(x,y)$ the transcript $\pi(x,y)$ is the concatenation of all the messages received during 
the execution of~$\pi$ on $(x,y)$.

\subparagraph*{Erasures.}
The communication channel connecting Alice and Bob is subject to \emph{erasure} noise. 
In each timestep, the channel accepts a symbol $s\in \Sigma$ and  outputs either $s$ or a special erasure mark ($\bot \notin \Sigma$). 
The noise is assumed to be worst-case (adversarial), where up to $T$ symbols may be replaced with erasure marks.
The value of $T$ is unbounded and \emph{unknown} to the parties---their protocol should be resilient to any possible $T\in \N$.
A \emph{noise pattern} is a  bit-string $E \in\{0,1\}^*$ that indicates erasures in a given execution instance of the protocol. Specifically, there is an erasure in the $i$-th timestep if and only if $E_i=1$. 
Given a specific instance where both parties terminate before its~$s$-th transmission, the number of erasures that are induced by $E$ on that instance is the Hamming weight of $E_1,\ldots,E_s$, i.e., the number of 1's in the bit-string. We sometime allow $E$ to be of infinite length, however our protocols will be resilient only against noise patterns with bounded amount of noise (i.e., if $E$ is infinite, then its suffix is required to be the all-zero string; we call such noise \emph{finite} noise). 
 
\subparagraph*{Coding scheme: Order of speaking, silence and termination.}
A \emph{coding scheme} is a protocol~$\Pi$ that takes as an input another protocol~$\pi$ (that assumes noiseless channels) and simulates~$\pi$ over a noisy channel.
By saying that $\Pi$ \emph{simulates}~$\pi$ over an erasure channel with (unbounded) $T$ corruptions, we mean the following.
For any pair of inputs~$(x,y)$ given to Alice and Bob,  
after executing $\Pi$ in the presence of $T$ erasures, 
Alice and Bob will produce the transcript~$\pi(x,y)$.
The above should hold for \emph{any}  $T\in \mathbb{N}$, independent of~$\pi,x,y$ and unknown to~$\Pi$ (i.e., to the parties).

We assume that the protocol~$\Pi$, at any given timestep, exactly one party can be the sender and the other party is the receiver.
That is, it is never the case that both parties send a symbol or both listen during the same timestep. 
On the other hand we do not assume that the parties terminate together, and it is possible that one party
terminates while the other does not. In this case, whenever the party that hasn't yet terminated is set to listen, it hears some default symbol~$\silence$, which we call \emph{silence}. 

There are several ways to treat silence. One option, taken in~\cite{AGS16}, is to treat
the silence similar to any other symbol of~$\Sigma$. That is, as long a party has not terminated and is set to speak, it can either send a symbol or remain silent (``send~$\silence$'') at that round. We take this approach in the scheme of Section~\ref{sec:adaptive}.
A different approach would be to require the parties to speak a valid symbol from~$\Sigma$ while 
they haven't terminated. That is, a party cannot
remain silent if it is set to speak; this prevents the parties from using silence as a means of
 communicating information during the protocol. 
Once a party terminates, and only then, silence is
being heard by the other party. 
This mechanism makes it easier for the parties to coordinate their termination.
In particular, the event of termination of one party transfers (limited) information to the other party. 
We take this approach in the scheme of Section~\ref{sec:scheme}.

Another subtlety %
stems from the fact that the length of the protocol 
is not predetermined.  
That is, the length of the protocol depends on the actual noise in the specific instance.
Such protocols are called adaptive~\cite{AGS16}.
In this case it makes sense to measure properties of the protocol
with respect to a specific instance. 
For instance, 
given a specific instance of the protocol $\Pi$ on inputs $(x,y)$ with some given noise pattern,
the communication complexity %
 $\CC(\Pi(x,y))$ is the number of symbols \emph{sent} by both parties in the specific instance. %
 The communication is usually measured in bits by multiplying the number of symbols by $\log|\Sigma|$.
The noise in a given instance is defined to be the number of corrupted transmissions until both parties have terminated, \emph{including} corruptions that occur after one party has terminated and the other party has not.
Corruptions made after both parties have terminated cannot affect the protocol, and we can assume such corruptions never happen.

%% file: sim.tex
\section{A coding scheme for an unbounded number of erasures}\label{sec:sim}\label{sec:scheme}
In this section we provide a coding scheme that takes as an input any noiseless protocol~$\pi$
and simulates it over a channel that suffers from an unbounded and unknown 
number of erasures~$T$. The coding scheme uses an alphabet~$\Sigma$ of size~4; in this setting parties are not allowed 
to be silent (unless they quit the protocol) and in every round they must send a symbol from~$\Sigma$.

\subsection{The Coding Scheme}
The coding scheme for Alice and Bob, respectively, is depicted in Algorithms~\ref{alg:sima} and~\ref{alg:simb}. 

Inspired by the simulation technique of~\cite{EGH16},
our simulation basically follows the behavior of the noiseless protocol~$\pi$ step by step, where Alice and Bob speak in alternating timesteps. 
In each timestep, the sending party tries to extend the simulated transcript by a single bit.
To this end, the parties maintain partial transcripts $\calT^{A}$ for Alice (and $\calT^{B}$ for Bob) 
which is the concatenation of the information bits of $\pi$ that the parties have simulated so far and are certain of. 
Then, if Alice is to send a message to Bob, she generates the next bit of $\pi$ given her partial simulated transcript, i.e., $\pi(x \mid \calT^{A})$, and sends this information to Bob.

In addition to the information bit, Alice also sends a parity of the round number she is currently simulating.
That is, Alice holds a variable $r_A$ which indicates the round number she is simulating. Recall that each
round contains two timesteps, where Alice communicates in the first timestep and Bob in the second.
Alice sends Bob the next bit according to her $\calT^{A}$ and waits for Bob's reply to see if this round was successfully simulated. If Bob's reply indicates the same parity (i.e., the same round), Alice knows her message arrived to Bob correctly and hence the round was correctly simulated. In this case Alice increases $r_A$. Otherwise, she assumes there was a corruption and she keeps $r_A$ as is; this causes the same round of~$\pi$ to be re-simulated in the next round of the simulation protocol.

Bob holds a variable $r_B$ which again holds the (parity of) the latest round in~$\pi$ he has simulated. In a somewhat symmetric manner (but not identical to Alice!), he expects to receive from Alice the bit of the \emph{next} round of~$\pi$, $r_B+1$. If this is the case he responds with his bit of that same round, or otherwise he re-transmits his bit of round $r_B$. 

Our coding scheme assumes a channel with alphabet $\Sigma = \{0,1\}\times\{0,1\}$, where every non-silent message can be interpreted as $m=(\Inf,\Par)$, where $\Inf \in \{0,1\}$ is the information bit (simulating~$\pi$) and $\Par \in \{0,1\}$ is the parity of the round of~$\pi$ simulated by the sender. 

The above continues until Alice has simulated the entire transcript of~$\pi$, i.e, when $r_A$ reaches the number of rounds in~$\pi$. At this point, Alice exits the protocol. %
Bob, however, cannot tell whether Alice has completed the simulation or not and waits until he sees a silence, which indicates that Alice has terminated, only then does he exit the protocol. As regards the correctness, we prove that in any round of the simulation, Bob has seen at least as much of the noiseless protocol that Alice herself has seen. In particular, when Alice exits, Bob must have seen the entire transcript of the noiseless protocol.

\IncMargin{1em}
\begin{algorithm}[h]
\caption{Simulation over Erasure Channel with Unbounded Noise (Alice's~side)}
\label{alg:sima}
\small
\DontPrintSemicolon
\KwData{An alternating binary protocol $\pi$ of length $N$ and an input $x$.}
\BlankLine
\nlset{A.1}\label{line:inita}\emph{Initialize $\calT^{A} \gets \emptyword$, $r_A \gets 0$. 
}\;
\nlset{A.2}\While{$(r_A < \frac{N}{2})$}{ 
\BlankLine
	\tcp*[l]{Send a Message (odd time-step)}
	\nlset{A.3}\label{line:inca}$r_A \gets r_A+1$\;
	\nlset{A.4}$b_{send} \gets \pi(x \mid \calT^{A})$\;
	\nlset{A.5}\label{line:apnda}$\calT^{A} \gets \calT^{A}\circ b_{send}$\;
	\nlset{A.6}\label{line:sa}send $(b_{send},r_A \mod 2)$\;

	\BlankLine\BlankLine
	\tcp*[l]{Receive a Message (even time-step)}
		\nlset{A.7}Obtain $m' = (b_{rec},r_{rec})$\;
		\nlset{A.8}\label{line:uncora}
		\eIf{$m' \neq \bot$ \textbf{and} $r_{rec}=r_A~\bmod 2$}
		{
			\nlset{A.9}$\calT^{A}\gets \calT^{A}\circ b_{rec}$\;
		}
		{
			\nlset{A.11}\label{line:dela}Delete last symbol of $\calT^{A}$\;
			\nlset{A.12}\label{line:deca1}$r_A \gets r_A-1$\; 
		}
}
\nlset{A.13}\label{line:outA}Output $\calT^{A}$\;
\end{algorithm}
\DecMargin{1em}

\IncMargin{1em}
\begin{algorithm}[h]
\caption{Simulation over Erasure Channel with Unbounded Noise (Bob's side)}\label{alg:simb}\DontPrintSemicolon
\small
\KwData{An alternating binary protocol $\pi$ of length $N$ and an input $y$.}
\BlankLine
\nlset{B.1}\label{line:initb}\emph{Initialize $\calT^{B} \gets \emptyword$, $r_B\gets 0$, $err \gets 0$, 
$m \gets (0,0)$}\;
\nlset{B.2}\While(\tcp*[f]{while Silence isn't heard}){$m' \neq \square$}{
\BlankLine
\tcp*[l]{Receive a Message (odd time-step)}
		\nlset{B.3}Obtain $m' = (b_{rec},r_{rec})$\;
		\nlset{B.4}\label{line:uncor}\uIf{$m' = \bot$\textbf{ or } $r = r_B~\bmod 2$}{
			\nlset{B.5}\label{line:err2}$err \gets 1$ \tcp*[f]{error detected}\;
		}
		\nlset{B.6}\Else{	%
			\nlset{B.7}$\calT^{B}\gets \calT^{B}\circ b_{rec}$ \;
			\nlset{B.8}\label{line:err0}$err\gets 0$\;
		}

\BlankLine\BlankLine
\tcp*[l]{Send a Message (even time-step)}
		\nlset{B.9}\uIf{$err = 0$}{
			\nlset{B.10}$b_{send} \gets \pi(y \mid \calT^{B})$\;
			\nlset{B.11}$\calT^{B} \gets \calT^{B}\circ b_{send}$\;
			\nlset{B.12}\label{line:incb}$r_B \gets r_B+1$\;
			\nlset{B.13}\label{line:sb1}$m \gets (b_{send},r_B \bmod 2)$\;				
		}
		\nlset{B.14}\Else (\tcp*[f] {$err=1$: keep $r_B$ and $m$ unchanged.}){
		\nlset{B.15}\label{line:sb2}send $m$\;
		}
}	
\nlset{B.16}\label{line:outB}Output $\calT^{B}$\;
\end{algorithm}
\DecMargin{1em}

%% file: correctness.tex
\subsection{Analysis}\label{sec:corr}

\subparagraph*{Preliminaries.} 
Recall that in our terminology a round consists of two timesteps, where at every timestep one party sends one symbol from~$\Sigma$. 
Alice sends symbols in odd timesteps and Bob in even ones. The above applies for both the noiseless protocol~$\pi$ (where the alphabet is binary) and for the coding scheme~$\Pi$ given by Algorithms~\ref{alg:sima} and~\ref{alg:simb} (where $\Sigma=\{0,1\}\times\{0,1\}$).

We think of the communication transcript as a string obtained by the concatenation of symbols sent during the course of the protocol run. 
Given a noiseless protocol~$\pi$ and inputs $x,y$, 
we denote by $m^{\pi}_{A}(i)$ (respectively, $m^{\pi}_{B}(i)$) the message sent by Alice (respectively, Bob) in the $i$-th round in the noiseless protocol~$\pi$. Let $\calT^{\pi}(r)$ be the transcript of the players after $r$ rounds in $\pi$. The length of~$\pi$ is denoted~$N$ and without loss of generality we will assume that~$N$ is even.

We start our analysis by fixing a run of the coding scheme~$\Pi$, 
specified by fixing an erasure pattern and inputs~$(x,y)$. 
Let $k$ be the number of timesteps in the given run, and note that $k$ is always odd.
We define \srs in the coding scheme $\Pi$ in the same way as we define rounds in the noiseless protocol. 
That is, for $i \in \{1,2,\ldots, \lceil\frac{k}{2}\rceil\}$, \sr $i$ (in the coding scheme) corresponds to the timesteps $2i-1$ and $2i$. Throughout this section by \sr $i$ we refer to the round in the coding scheme unless it is specified otherwise. 
For the sake of clarity, we maintain that \sr $i$ begins at the start of the $(2i-1)$-th timestep and ends at the ending of the $2i$-{th} timestep. Since $\Pi$ is alternating, Alice begins a \sr by sending a message and the \sr ends after Alice receives Bob's response. 

We define $t_A$ (and $t_B$)
to be the termination round of Alice (and respectively, Bob). 
Observe that $t_A <  t_B = \lceil\frac{k}{2}\rceil$ since Bob terminates only once he hears silence, which can happen only if Alice has already terminated in a previous round. 
For any round~$i$ and variable $v$ 
we denote by $v(i)$ the value of~$v$ at the beginning of round~$i$.
In particular, 
$r_A(i)$ and $r_B(i)$ are the values of $r_A$ and $r_B$ at the beginning of round~$i$.
Since $r_A$ is not defined after Alice exits in Algorithm~\ref{alg:sima},  
for all $t_A< i\leq t_B$, we set $r_A(i)$ to be the value of~$r_A$ at the \emph{end} of round $t_A$, that is, its value just before Alice quits. We similarly define $\calT^{A}(i)$ for $t_A < i\leq t_B$. 
Finally, let $m_A(i)=(b_A(i), \rho_A(i))$ and $m_B(i)=(b_B(i), \rho_B(i))$ be the messages sent by Alice and Bob (see Lines~\ref{line:sa} or~\ref{line:sb1} or~\ref{line:sb2}) respectively in \sr $i$.

\newpage

\subparagraph*{Technical lemmas and proof of correctness.} 
The main technical claim of this part is Lemma~\ref{lem:shiftsync}, which we will now prove. 
We begin with the simple observation that, in every round, the parties' transcripts (and the respective round number the parties believe they simulate) either increase by exactly the messages exchanged during the last round, or they remain unchanged. 
\begin{lemma}\label{lem:nd}
For any $i \in [t_B]$, the following holds.
\begin{enumerate}
\item $r_A(i+1) \in \{r_A(i),r_A(i)+1\}$ and $\calT^{A}(i+1) \in \left\{\calT^{A}(i) , \calT^{A}(i)\circ b_A(i)\circ b_B(i)\right\}$ and, 
\item $r_B(i+1) \in \{r_B(i),r_B(i)+1\}$ and $\calT^{B}(i+1) \in \left\{\calT^{B}(i),\calT^{B}(i)\circ b_A(i)\circ b_B(i)\right\}$. 
\end{enumerate}
Furthermore, $r_A$ ($r_B$) changes if and only if $\calT^{A}$ ($\calT^{B}$) changes.
\end{lemma} 
\begin{proof}
Consider Algorithm~\ref{alg:simb}. It is immediate that $r_B$ can increase by at most~1 in every round.
In \sr $i$ Bob starts by either appending $b_{rec}=b_A(i)$ to $\calT^{B}(i)$ and setting $err=0$; or keeping $\calT^{B}(i)$ unchanged and setting $err=1$. In the former case, Bob appends $b_{send}=b_B(i)$ to $\calT^{B}(i)\circ b_A(i)$ and increases $r_B$. 
Otherwise, he keeps $\calT^{B}(i)$ unchanged, and $r_B$ remains the same as well.

In Algorithm~\ref{alg:sima} Alice may decrease her $r_A$, but note that she always begins \sr $i \le t_A$ by increasing it (Line~\ref{line:inca}) and appending $b_{send}=b_A(i)$ to $\calT^{A}(i)$ (Line~\ref{line:apnda}). She then either decreases $r_A$ back to what it was and in this case she erases $b_A(i)$ (see Lines~\ref{line:deca1} and~\ref{line:dela}), or she keeps it incremented and appends $b_{rec}=b_B(i)$ to $\calT^{A}(i)\circ b_A(i)$.  For $i>t_A$, i.e., after Alice terminates, the above trivially holds.
\end{proof}
\begin{corollary}\label{cor:trnoinc}
For any two \srs $i\leq j$, $\calT^{A}(i)=\calT^{A}(j)$ if and only if $r_A(i)=r_A(j)$. Similarly, $\calT^{B}(i) = \calT^{B}(j)$ if and only if $r_B(i)=r_B(j)$.
\end{corollary}
\begin{proof}
Note that $\calT^{A},r_A$ are non-decreasing. Lemma~\ref{lem:nd} proves that these two variable increase simultaneously, which proves this corollary. The same holds for $\calT^{B},r_B$.
\end{proof}
\goodbreak
\begin{lemma}\label{lem:shiftsync}
For $i\in [t_A+1]$, one of the following conditions holds:
\begin{enumerate}
\item\label{it:one}$r_B(i) = r_A(i)$ and $\calT^{B}(i) = \calT^{A}(i)$ or,
\item\label{it:two}$r_B(i) = r_A(i)+1$ and $\calT^{B}(i) = \calT^{A}(i)\circ b_A(i-1)\circ b_B(i-1)$.
\end{enumerate}
Furthermore, in either case, $\calT^{A}$ and $\calT^{B}$ are prefixes of $\calT^{\pi}$.%
\end{lemma}
\begin{proof}
We prove the lemma by induction on the round~$i$. 
In the first \sr, $r_A(1) = r_B(1)=0$ and $\calT^{B}(1)=\calT^{A}(1)=\emptyword$ (see Lines~\ref{line:inita},~\ref{line:initb}), thus Item~\ref{it:one} is satisfied. Note that $\calT^{A}$ and $\calT^{B}$ are trivially the prefixes of $\calT^{\pi}$. 
Now assume that the conditions hold at (the beginning of) round~$i\le t_A$ and consider what happens during this round. Note that both parties are active during round~$i$.

\begin{description}
\item [Case 1: $r_A(i) = r_B(i)$.]
Suppose Alice receives an (uncorrupted) message from Bob that carries the parity $r_A(i)+1$  (i.e., when Alice executes Line~\ref{line:uncora} and not the else part of Lines~\ref{line:dela}--\ref{line:deca1}). 
This means that Bob sends a message with parity $r_A(i)+1= r_B(i)+1$. Since $r_B(i)=r_A(i)$, this can only happen if Bob executed Lines~\ref{line:sb1} and~\ref{line:incb}, hence, $r_B(i+1) = r_B(i)+1$. %
If Alice receives a corrupted message, she decrements $r_A$ (that she had increased at the beginning of the round) and also deletes the last message from her transcript. Bob however may have received Alice's message correctly and in that event, he will increment his value of $r_B$ and update his transcript.

Note that if $r_B(i+1) = r_A(i+1)$, then by the induction hypothesis and Lemma~\ref{lem:nd} both transcripts either remained the same in round~$i$ (so they are still the same at the beginning of round~$i+1$), or both transcripts increased by appending $b_A(i)\circ b_B(i)$ to each, so they are still equal. Here, $b_A(i)=\pi(x \mid \calT^{A}(i))$ and $b_B(i)= \pi(y \mid \calT^{B}(i)\circ b_A(i))$. Similarly, if $r_B(i+1) = r_A(i+1)+1$, Lemma~\ref{lem:nd} establishes that $\calT^{A}(i)$ hasn't changed in round~$i$, while $\calT^{B}(i+1) = \calT^{B}(i)\circ b_A(i)\circ b_B(i)$, which by the induction hypothesis gives Item~\ref{it:two}. %
Since $b_A(i)$ and $b_B(i)$ are the correct continuations of $\calT^{A}(i)$ from $\pi$, the above discussion proves that $\calT^{A}(i+1)$ and $\calT^{B}(i+1)$ are prefixes of $\calT^{\pi}$.

\item [Case 2: $r_B(i)=r_A(i)+1$.] 
In this case, whether Bob receives an erasure or an uncorrupted message, he sets $err=1$.
Indeed, if Alice's message is not erased, then the parity Bob receives equals his saved parity 
(since Alice holds $r_A(i)=r_B(i)-1$ and she increases it by one in Line~\ref{line:inca} before sending it to Bob).
In both cases Bob does not change $r_B$, i.e.,  $r_B(i+1)=r_B(i)$ and he sends the message $m_B(i) = (b_B,r_B~\bmod 2)$ (see Lines~\ref{line:err2} and~\ref{line:sb2}) from his memory. Consequently, Bob's transcript doesn't change in round $i$. So, $\calT^{B}(i+1)=\calT^{B}(i)=\calT^{A}(i)\circ b_A(i-1)\circ b_B(i-1)$.

If Alice receives an erasure she sets $r_A(i+1)$ to be same as $r_A(i)$ (see Line~\ref{line:deca1}) and sets $\calT^{A}(i+1)=\calT^{A}(i)$. Since, $r_A(i+1) = r_A(i)$ and $r_B(i+1) = r_B(i)$ the claim holds. However, if Alice receives an uncorrupted message, she notices that $r_{rec}  = (r_A(i) + 1) \bmod 2$ (see Line~\ref{line:uncora}) and she does not decrement $r_A$. In this case, $r_A(i+1) = r_A(i)+1= r_B(i) = r_B(i+1)$ and $\calT^{A}(i+1)=\calT^{A}(i)\circ b_A(i)\circ b_B(i)$.
Now we prove that $b_A(i)=b_A(i-1)$ and $b_B(i)=b_B(i-1)$, thus in the former case, $\calT^{B}(i+1) = \calT^{A}(i+1)\circ b_A(i)\circ b_B(i)$ or $\calT^{A}(i+1)=\calT^{B}(i+1)$. in the latter. 
Since in round $i$, Bob sets $err=1$ we have that $m_B(i)=m_B(i-1)$ whence $b_B(i)=b_B(i-1)$.
To prove that the same holds for $b_A$ we will need the following simple claim.
\begin{claim}\label{clm:ra_unchanged} 
$r_A(i)=r_A(i-1)$.
\end{claim}

\begin{proof}
Supposing round $i-1$ satisfies Item~\ref{it:one}, we must have that $r_A(i)=r_A(i-1)$. If this were not true, then 
$r_B(i)=r_A(i)+1=(r_A(i-1)+1)+1=r_B(i-1)+2$ which is a contradiction of Lemma~\ref{lem:nd}. 
On the other hand, if round $i-1$ satisfies Item~\ref{it:two}, then we know $r_B$ cannot increase in round $i-1$, whence, $r_B(i)=r_B(i-1)$. Using this we have, $r_A(i-1)+1=r_B(i-1)=r_B(i)=r_A(i)+1$. 
\end{proof}
Following the above claim, $r_A(i-1)=r_A(i)$, and we get that 
$\calT^{A}(i-1)=\calT^{A}(i)$ (Lemma~\ref{lem:nd}). 
Therefore, $b_A(i-1)=\pi(x\mid \calT^{A}(i-1))=\pi(x\mid \calT^{A}(i))=b_A(i)$.
From the induction hypothesis, $\calT^{B}(i+1)=\calT^{B}(i)$ is a prefix of $\calT^{\pi}$. From the above discussion, we know that either $\calT^{A}(i+1)=\calT^{A}(i)\circ b_A(i)\circ b_B(i)=\calT^{B}(i+1)$ or $\calT^{A}(i+1)=\calT^{A}(i)$. In the former case, it is clear that $\calT^{A}(i+1)$ is a prefix of $\calT^{\pi}$ whereas in the latter case, since $\calT^{A}(i)$ is a prefix of $\calT^{\pi}$ (by the induction hypothesis) so is $\calT^{A}(i+1)$. \qedhere
\end{description}
\end{proof}
The following lemma implies that after Alice terminates, no matter what erasures Bob sees, his values of $r_B$ and $\calT^{B}$ do not change.
\begin{lemma}\label{lem:end}
For~$i$ such that $t_A < i\leq t_B$, \sr $i$ satisfies Item~\ref{it:one} of Lemma~\ref{lem:shiftsync} and $r_A(i)=\frac{N}{2}$.
\end{lemma}
\begin{proof}
Recall that $t_A < t_B$. 
Since Alice exits in \sr $t_A$, it must hold that at the end of this round $r_A=\frac{N}{2}$, yet,
$r_A(t_A) = \frac{N}{2}-1$ for otherwise, Alice would have terminated in the end of \sr $t_A-1$. 
Via Lemma~\ref{lem:shiftsync} we know that $r_B(t_A) \in \{\frac{N}{2}-1,\frac{N}{2}\}$. If $r_B(t_A)=\frac{N}{2}-1$ we know from the proof of case~1 in Lemma~\ref{lem:shiftsync} that since $r_A$ increases in round $t_A$, $r_B$ also increases. 
In the other case, namely, when $r_B(t_A)=\frac{N}{2}$, 
we know from the proof of 
Lemma~\ref{lem:shiftsync} (case~2) that $r_B$ does not change, i.e., 
$r_B(t_A+1)=r_B(t_A) = \frac{N}{2}$. 

After Alice exits, Bob can either hear silence or an erasure and therefore for all \srs $i > t_A$, Bob sets $err=1$ and consequently $r_B(i)$ is never incremented. It follows that for all $i > t_A$, $r_B(i) = \frac{N}{2}$. The second part of the claim follows from Corollary~\ref{cor:trnoinc} and Lemma~\ref{lem:shiftsync}, since $\calT^{A}(t_A+1) = \calT^{B}(t_A+1)$ and since $r_A,r_B$ do not change anymore.
\end{proof}

We are now ready to prove the main theorem and show that we can correctly simulate the noiseless protocol~$\pi$ under the specifications of Theorem~\ref{thm:main}. We first state Lemma~\ref{lem:progress} which will be help us bound the communication and also show that Alice eventually terminates.

\begin{lemma}\label{lem:progress}
In any round $i\in [t_A]$ where there are no erasures at all, $r_A(i+1)=r_A(i)+1$.
\end{lemma}
\begin{proof}
Note that both parties are still active in round~$i$.
Consider the two cases of Lemma~\ref{lem:shiftsync}. If $\calT^{A}(i)=\calT^{B}(i)$ and both messages of round $i$ are not erased, we showed that $r_A$ increases (case~1).
Similarly, if $\calT^{B}(i) = \calT^{A}(i)\circ b_A(i-1)\circ b_B(i-1)$ and no erasures occur, Alice extends her transcript and $r_A$ increases again (case~2).
\end{proof}

\begin{theorem}\label{thm:final}
Let $\pi$ be an alternating binary protocol and $T\in \mathbb{N}$ be an arbitrary integer. 
There exists a coding scheme $\Pi$ over a $4$-ary alphabet such that for any instance of~$\Pi$ that suffers at most $T$ erasures overall, Alice and Bob both output $\calT^{\pi}$. 
The simulation~$\Pi$ communicates at most $\CC(\pi)+2T$ symbols, and has $\CC(\Pi)\leq 2\CC(\pi)+4T$.  
\end{theorem}
\begin{proof}

Lemma~\ref{lem:shiftsync} guarantees that at every given round, Alice and Bob hold a correct prefix of~$\calT^{\pi}$. Moreover, we know that by the time Alice terminates, her transcript (and hence Bob's transcript) is of length at least~$N$, which follows from Lemma~\ref{lem:end} and Lemma~\ref{lem:nd}, i.e., from the fact that $r_A=N/2$ at termination, and that every time $r_A$ increases by one, the length of $\calT^{A}$ increases by two. Finally, note that if the number of erasures is bounded by $T$, then Alice will eventually reach termination because, after $T$ erasures, $r_A$ increases by one in every round (Lemma~\ref{lem:progress}), until it reaches $N/2$ and Alice terminates. 

Finally, we need to prove that the communication behaves as stated. 
Assume $T'$ erasures happen up to round~$t_A$ and $T''=T-T'$ erasures happen in rounds $[t_A+1,t_B]$.
Since every round without erasures advances $r_A$ by one (again from Lemma~\ref{lem:progress}), and since when Alice terminates we have $r_A=N/2$ (Lemma~\ref{lem:end}), then $t_A \le T'+N/2$. Furthermore, after Alice terminates it takes one unerased (odd) round to make Bob terminate as well. Hence, at every round after $t_A$ and until Bob terminates, Alice's silence must be erased. It follows that $t_B-t_A \le 1+T''$.
Thus, $t_B \le 1+ T'' + t_A \le 1+ T + N/2$. 

Every round $i\in[t_A]$ consists of two transmissions, while every round $t_A <i < t_B$ contains only a single transmission---Bob's transmission, excluding round $t_B$ where Bob hears silence and terminates without sending a message. The total number of transmissions is then,
\[
t_A + t_B - 1 \le N + 2T' + T'' \le N+2T.
\]

Recall that $|\Sigma|=4$, hence 
$\CC(\Pi) \le (N + 2T)\log 4 = 2(\CC(\pi)+2T)$.
\end{proof}

\subsection{Noise Resilience and Code Rate} 
We can compare the above result to the case where the noise is bounded as a fraction of the symbols communicated in the protocol~\cite{schulman96,FGOS15,EGH16}. 
In our setting, the noise amount can be arbitrary. In order to compare it to the bounded-noise model,
we ask the following question. Assume an instance of $\Pi$ with large amount of noise~$T$ ($T\gg N)$.
Then, what fraction does the noise make out of the entire communication.

As a corollary of Theorem~\ref{thm:final} it is easy to see that the fraction of noise is lower bounded by $\frac{T}{N+2T}$ whose limit, as $T$ tends to infinity, is $1/2$. Indeed, $1/2$ is an upper bound on the noise fraction in the bounded-noise setting~\cite{FGOS15}.\footnote{Note that a noise level of~$1/2$ is also achievable for interactive coding over erasure channels in the bounded-noise setting, for alphabets of size at least~$4$~\cite{EGH16,FGOS15}. The maximal noise for erasure channels with \emph{binary} or \emph{ternary} alphabet is still open.} 
Furthermore, if the noise is bounded to be a $\delta$-fraction of the total communication, for some $\delta< 1/2$ then, $T \le \delta \cdot 2t_B \le \delta(N+2T)$ and so $T\le \frac{\delta N}{1-2\delta }$. This implies a maximal asymptotic code rate of $1/2$. Indeed,
\[R = \frac{\CC(\pi)}{\CC(\Pi)} \ge \frac{1}{\log|\Sigma|}\cdot\frac{N}{N+2T} \ge \frac{1}{2}\cdot\frac{N}{N+2\frac{\delta N}{1-2\delta }}\ge \frac{1}{2}(1-2\delta) = \frac12 - \delta.\]
As $T$ is unbounded relative to $\CC(\pi)$ and we can potentially get a zero rate, a more interesting measure is the ``waste'' factor, i.e., how much the communication of $\Pi$ increases per single noise, for large $T$. In our scheme it is easy to see that each corruption delays the simulation by one round, that is, it wastes two symbols (4 bits). 
This implies a waste factor of~$4$ bits per corruption,
$\omega =\lim_{T\to \infty}\frac{\CC(\Pi)}{T} = \frac{2N+4T}{T} = 4$.

Finally, we mention that
our result extends to binary alphabet by naively encoding each symbols as two bits (this also proves the second part of Theorem~\ref{thm:main}).
However, this results in a reduced tolerable noise rate of~$\tfrac14$. 
Similar to the scheme in~\cite{EGH16}, the noise resilience can be improved to~$\tfrac13$
by encoding each symbol via an error correcting code of cardinality~$4$ and distance~$\frac{2}{3}$, e.g.,
\{000,011,110,101\}. In this case, two bits must be erased in order to invalidate a round.
Each round (two timesteps) consists the transmission of six bits. Hence, the obtained resilience is $2/6=1/3$, similar to the best known resilience in the bounded-noise setting with binary alphabet~\cite{EGH16}.

%% file: sim3.tex
\section{A coding scheme in the AGS adaptive setting}\label{sec:adaptive}
In this section we provide an adaptive coding scheme in the AGS setting
that simulates any noiseless protocol~$\pi$ and is resilient to an unbounded and unknown 
amount of erasures~$T$.

The setting of this section is based on the adaptive setting described in Section~\ref{sec:model}
with the following main difference: at any given round, parties are allowed to either remain silent
\emph{or} send a symbol from $\Sigma$. The other party is assumed to listen and it either
hears silence, a symbol from~$\Sigma$, or an erasure mark in case the channel corrupted the transmission.

The (symbol) communication complexity of a protocol in this setting, $\CC^{sym}(\Pi)$, is defined to be the number symbols the parties communicate, i.e., the number of 
timesteps in which the sender did not remain silent and transmitted a symbol from~$\Sigma$. 
We begin with an alphabet of size~4 and then in Section~\ref{sec:unary} we discuss how to reduce the alphabet to being unary---either a party speaks (sends energy) or it remains silent (no energy). 
The (symbol) communication complexity then portrays the very practical quantity of ``how much energy'' the protocol costs. In this case,  
we define $\CC(\Pi) = \CC^{sym}(\Pi)$.

Another difference from the previous setting regards termination. In this section, we do not require the parties to terminate and give output. Instead we only require that at some point the parties compute the correct output and that the communication is bounded. Formally, this property of semi-termination is defined as follows.
\goodbreak
\begin{definition}\label{def:semi-term}
We say that an adaptive protocol has semi-termination if there exists a round $t$ after which both the following conditions hold:
\begin{enumerate}
\item Both parties have computed the correct output and,
\item Both parties remain silent indefinitely (whether they terminate or not)
\end{enumerate}
\end{definition}

The round complexity of an instance, $\RC(\Pi)$ is the number of timesteps the instance takes until 
the parties terminate. In case of semi-termination, $\RC(\Pi)$ is the number of timesteps until the point in time $t$ where the semi-termination conditions hold according to Definition~\ref{def:semi-term}.

\subsection{The coding scheme}\label{sec:adaptive:coding}
The adaptive coding scheme for this setting is described in Algorithms~\ref{alg:sima1} and~\ref{alg:simb1}. We now give an overview and  in Section~\ref{sec:adaptive:analysis} we prove the correctness of the scheme and analyze its properties such as communication complexity, etc. As mentioned earlier, Algorithms~\ref{alg:sima1} and~\ref{alg:simb1}
assume a channel alphabet of size~$4$ whose size we reduce in Section~\ref{sec:unary}.

The coding scheme is based on the scheme of Section~\ref{sec:sim}.
Similar to the scheme in Section~\ref{sec:sim}, 
we maintain that Alice and Bob communicate only in odd and even timesteps, respectively.  
Therefore, in odd (even) timesteps, Alice (Bob) may either send a message or remain silent. 
Whenever a player chooses to speak, they will send a message of the form $(\Inf, \Par) \in \{0,1\}\times\{0,1\}$. Here, $\Inf$ is the information bit specified by the noiseless protocol ($\pi$), based on the input and the information received by the player so far. Formally, Alice and Bob maintain the ``received information'' in the form of a partial transcript---$\calT^{A}$ or $\calT^{B}$---which is the concatenation of the information bits exchanged until that point in the protocol. 
$\Par$ is the parity of the round number (of $\pi$) that the player is currently simulating. Recall that a round corresponds to two timesteps where Alice transmits the first message and Bob the second message. As before, Alice maintains her round number as $r_A$ and Bob his as $r_B$.

The main difference between the coding scheme in this section and the previous regards the way parties signal that an erasure has happened. Here, whenever Alice or Bob receive an erasure, they simply remain silent.
This signals the other side that there was an erasure and that the last message should be re-transmitted.
Otherwise, they behave similar to the scheme in Section~\ref{sec:sim}. 
Namely,
when they receive a non-silence non-erased message 
 $(\Inf,\Par)$, they check the received parity: 
 if it corresponds to the next-expected round they update their partial transcript $\calT$ by appending the new information bit, and increase the round number~$r$ by one. If the parity indicates a mismatch to the current state, they just re-transmit the last message again (as they do in case they hear silence).

The scheme in this section has a semi-termination property, where Bob never terminates the protocol while Alice terminates once her partial transcript obtained the required length. For Bob, it is possible that  his partial transcript reaches the maximal length while Alice has not yet completed the simulation. Thus, 
Bob never exits and instead it remains silent waiting for a message from Alice. If indeed Alice hasn't completed the simulation such a message is bound to arrive sometime, and Bob will reply to it with the last message of the protocol. Otherwise, Bob will keep waiting. 
To achieve this behavior,
once Bob's transcript reaches the maximal length he transitions to what we call \emph{termination phase}.
In termination phase, Bob no longer responds to silence by re-transmitting his message, but instead he responds to silence with silence.

\IncMargin{1em}
\begin{algorithm}[ht]
\linespread{1}\selectfont
\small
\caption{Simulation in the AGS setting (Alice's side)}\label{alg:sima1}\DontPrintSemicolon
\KwData{An alternating binary protocol $\pi$ of length $N$ and an input $x$.}
\BlankLine
\nlset{A.1}\label{line:adap:inita}\emph{Initialize $\calT^{A} \gets \emptyword$, $r_A\gets 0$ and $m_{rec}\gets (0,0)$}\;
\nlset{A.2}\label{line:adap:AliceExit}\While{$r_{A} < \frac{N}{2}$}{

	\tcp*[l]{Odd timesteps}
	\nlset{A.3}\label{line:adap:inca}$r_A \gets r_A+1$\;
	\nlset{A.4}\label{line:adap:aif1}\uIf{$m_{rec}\neq \bot$}{
		\nlset{A.5}\label{line:adap:senda}$m_{send}\gets (\pi(x \mid \calT^{A}),r_A\bmod 2))$\;
		}
	\Else{
		\nlset{A.6}$m_{send}\gets \square$ \;
	}
	\nlset{A.7}send $m_{send}$\;
	\BlankLine
	\BlankLine
	\tcp*[l]{Even timesteps}
	\nlset{A.8}receive $m_{rec} = (b_{rec},r_{rec})$\;
	\nlset{A.9}\label{line:adap:aif2}\uIf{$r_{rec}= r_A\bmod 2$}{
		\nlset{A.10}\label{line:adap:appa}$\calT^{A}\gets \calT^{A}\circ \pi(x \mid \calT^{A})\circ b_{rec}$\;
	}
	\Else{	
		\nlset{A.11}\label{line:adap:deca}$r_A \gets r_A-1$\;	
	}
}
\nlset{A.12}Output $\calT^{A}$\;
\end{algorithm}
\DecMargin{1em}

\IncMargin{1em}
\begin{algorithm}[ht]
\linespread{1}\selectfont
\caption{Simulation in the AGS setting (Bob's side)}\label{alg:simb1}\DontPrintSemicolon
\small
\KwData{An alternating binary protocol $\pi$ of length $N$ and an input $y$.}
\BlankLine
\nlset{B.1}\label{line:adap:initb}\emph{Initialize $\calT^{B} \gets \emptyword$, $b_{send}\gets0$, $r_B\gets 0$, and $m_{rec}\gets (0,0)$}\;
\BlankLine
\nlset{B.2}\While{$r_B < \frac{N}{2}$}{

	\tcp*[l]{Odd timesteps}
	\nlset{B.3}receive $m_{rec} = (b_{rec},r_{rec})$\;
	\nlset{B.4}\label{line:adap:bif2}\If{$r_{rec}= r_B+1\bmod 2$}{
		\nlset{B.5}\label{line:adap:incb}$r_B \gets r_B+1$\;	
		\nlset{B.6}\label{line:adap:setb}$b_{send} \gets \pi(y\mid \calT^{B}\circ b_{rec})$\;
		\nlset{B.7}\label{line:adap:appb}$\calT^{B}\gets \calT^{B}\circ b_{rec}\circ b_{send}$\;
	}
	\BlankLine\BlankLine

	\tcp*[l]{Even timesteps}
	\nlset{B.8}\label{line:adap:bif1}\uIf{$m_{rec}\neq \bot$}{
		\nlset{B.9}\label{line:adap:sendb} $m_{send} \gets (b_{send},r_B~\bmod 2)$\;
		}
	\Else{
		\nlset{B.10}\label{line:adap:bsendsq}$m_{send}\gets \square$\;
	}
	\nlset{B.11}send $m_{send}$\;
}
\BlankLine\BlankLine
\tcp*[l]{Once $r_B$ reached $N/2$ switch to termination phase}
\nlset{B.12}\label{line:adap:terminationloop}\While{$r_B = \frac{N}{2}$}{
	\tcp*[l]{Odd timesteps}
	\nlset{B.13}receive $m_{rec} = (b_{rec},r_{rec})$\;
	\BlankLine\BlankLine
	
	\tcp*[l]{Even timesteps}
	\nlset{B.14}\label{line:adap:bif3}\uIf{$r_{rec}= r_B\bmod 2$}{
		\nlset{B.15}\label{line:adap:bobLast}$m_{send} \gets (b_{send},r_B\bmod~2)$\;
	}
	\Else
	{	\nlset{B.16}$m_{send} \gets \square$\;}
	\nlset{B.17}send $m_{send}$\;
}
\nlset{B.18}The output is $\calT^{B}$\;    
\end{algorithm}
\DecMargin{1em}

%% file: corr3.tex
\subsection{Analysis}\label{sec:adaptive:analysis}
We now analyze the coding scheme~$\Pi$ described in Algorithms~\ref{alg:sima1} and~\ref{alg:simb1} and prove that it (1) simulates any $\pi$ even in the presence of~$T$ erasures (for any a priori unknown~$T$), and (2) the simulation communicates  $\CC^{sym}(\Pi)\leq \CC(\pi)+T$ non-silent transmissions. 
The above properties are formulated as Theorem~\ref{thm:final1} at the end of this subsection.

The approach of the analysis goes in large parts 
similar to the analysis of Algorithms~\ref{alg:sima} and~\ref{alg:simb} which appeared in Section~\ref{sec:corr}. 
First we show that at any given round, the difference between the round numbers maintained by Alice and Bob is at most one. Furthermore, Bob's partial transcript is always as long as the one of Alice, but no longer than Alice's partial transcript plus the information bits corresponding to the ongoing round. 
Then, we prove that the parties' partial transcript is always the correct one and that for any round in which there are no erasures---the length of the partial transcript increases.
It then follows that when Alice terminates the protocol, we are guaranteed that her partial transcript is the entire simulation. Since Bob's partial transcript can only be longer, this means he also correctly simulated $\pi$.

We start our analysis by fixing a run of the coding scheme $\Pi$.
We recall that round~$i$ (of~$\Pi$) corresponds to timesteps $(2i-1)$ and $2i$ 
and that  Alice and Bob send symbols in odd and even timesteps respectively. Since $\Pi$ is alternating, Alice begins a \sr by sending a message and the \sr ends after Alice receives Bob's response.
We define $t_A$
to be the termination round of Alice. 
Observe that $t_B=\infty$, i.e., Bob never terminates. 
Indeed, once Bob enters the second while loop of Algorithm~\ref{alg:simb1} (Line~\ref{line:adap:terminationloop}) 
he never changes his value of $r_B$ and continues to execute this loop. 
We borrow the definitions of $\calT^{A}(i),\calT^{B}(i),r_A(i)$ and $r_B(i)$ from Section~\ref{sec:corr}. Also, let $|\calT^{B}(i)|$ be the length of the string~$\calT^{B}(i)$.
Since $r_A$ is not defined after Alice exits in Algorithm~\ref{alg:sima}, we set $r_A(i)=r_A(t_A+1)$ for all~$i > t_A$.
Similarly, we set $\calT^{A}(i)=\calT^{A}(t_A+1)$ for all~$i > t_A$. 
If in round $i$ Alice or Bob do not remain silent then $m_A(i)=(b_A(i), \rho_A(i))$ and $m_B(i)=(b_B(i), \rho_B(i))$ be the messages sent by Alice or Bob, respectively (Line~\ref{line:adap:senda} and Line~\ref{line:adap:sendb}). In case Alice or Bob remain silent then we set $m_A(i)$ (or $m_B(i)$) to be $\square$ and $b_A(i)$ (or $b_B(i)$) to be the empty word~$\emptyword$.

\begin{lemma}[Analog of Lemma~\ref{lem:nd}]\label{lem:nd1}
For any $i \geq 0$, the following holds.
\begin{enumerate}
\item\label{it:1} $r_A(i+1) \in \{r_A(i),r_A(i)+1\}$ and $\calT^{A}(i+1) \in \left\{\calT^{A}(i),\calT^{A}(i)\circ \pi(x\mid \calT^{A}(i))\circ b_B(i)\right\}$ and, 
\item\label{it:2} $r_B(i+1) \in \{r_B(i),r_B(i)+1\}$ and $\calT^{B}(i+1) \in \left\{\calT^{B}(i),\calT^{B}(i)\circ b_A(i)\circ \pi(y\mid \calT^{B}(i) \circ b_A(i))\right\}$. 
\end{enumerate}
Furthermore, $r_A$ ($r_B$) changes if and only if $\calT^{A}$ ($\calT^{B}$) changes.
\end{lemma} 
\begin{proof}
In Algorithms~\ref{alg:sima1} and~\ref{alg:simb1} we note that $r_A$ and $r_B$ respectively can increase by at most one in every round. 
From Lines~\ref{line:adap:incb}--\ref{line:adap:appb} it is clear that $\calT^{B}$ changes (according to Item~\ref{it:2}) if and only if $r_B$ changes. 
We note that Alice begins every round by incrementing $r_A$ (Line~\ref{line:adap:inca}) in the odd timestep. From Lines~\ref{line:adap:aif2}--\ref{line:adap:deca}, we see that $\calT^{A}$ changes (according to Item~\ref{it:1}) if and only if Alice does not decrement $r_A$ in the even timestep.
Lastly, whenever $\calT^{A}$ or $\calT^{B}$ changes it is clear from Lines~\ref{line:adap:appa} and~\ref{line:adap:appb} that $\calT^{A}(i+1)=\calT^{A}(i)\circ \pi(x\mid \calT^{A}(i))\circ b_B(i)$ and $\calT^{B}(i+1)=\calT^{B}(i)\circ b_A(i)\circ \pi(y\mid \calT^{B}(i) \circ b_A(i))$ respectively. 

Once Alice exits (or Bob enters termination phase), the values of $r_A,\calT^{A}$ (or $r_B,\calT^{B})$ never change.
\end{proof}

\begin{lemma}[Analog of Lemma~\ref{lem:shiftsync}]\label{lem:shiftsync1}
For $i\in [t_A+1]$, one of the following conditions holds:
\begin{enumerate}
\item\label{it:one}$r_B(i) = r_A(i)$ and $\calT^{B}(i) = \calT^{A}(i)$ or,
\item\label{it:two}$r_B(i) = r_A(i)+1$ and 
	$\calT^{B}(i) = \calT^{A}(i)\circ \pi(x\mid \calT^{A}(i))\circ \pi\left(y\mid \calT^{A}(i)\circ\pi\left(x\mid \calT^{A}(i)\right)\right)$.
\end{enumerate}
Furthermore, in either case, $\calT^A$ and $\calT^B$ are prefixes of $\calT^{\pi}$.
\end{lemma}
\begin{proof}
We prove the lemma by induction on the round~$i$. 
In the first round, we initialize $r_A(1) = r_B(1)=0$ and $\calT^{B}(1)=\calT^{A}(1)=\emptyword$ (Lines~\ref{line:adap:inita},~\ref{line:adap:initb}), thus Item~\ref{it:one} is satisfied and moreover the transcripts are trivially correct prefixes of~$\calT^{\pi}$. Now assume that the conditions hold at (the beginning of) round $j$ for all $j\leq i$, and consider what happens during round~$i\le t_A$; note that Alice is active in this round. 

\begin{description}
\item [Case 1: $r_A(i) = r_B(i)$.]
If Alice does not increase $r_A(i)$ in this round, then either Item~\ref{it:one} or~\ref{it:two} 
hold via Lemma~\ref{lem:nd1}. Note that $\calT^{A}(i)=\calT^{B}(i)$ is a correct prefix of~$\calT^{\pi}$ by induction and Alice sends $b_A(i) = \pi(x\mid \calT^{A}(i))$ (Line~\ref{line:adap:senda}).

Now assume that Alice increases $r_A$ at the end of round $i$, i.e., she  
executes the \emph{if} part (Line~\ref{line:adap:appa}) rather than the \texttt{else} part (Line~\ref{line:adap:deca}). Therefore, Bob must have sent a message with parity $r_A(i)+1\equiv r_B(i)+1 \mod  2$. 
Since $r_B(i) \ne r_A(i)+1 \mod 2$, 
the above can only happen if Bob has executed Lines~\ref{line:adap:incb}--\ref{line:adap:appb} and increased his own $r_B$.  Hence, $r_A(i+1) = r_B(i+1)$.

Lemma~\ref{lem:nd1} indicates that both partial transcripts, $\calT^A$ and $\calT^B$ have increased during round $i$. We are left to show that they are identical and correct prefixes of~$\calT^{\pi}$. 
Indeed, both transcripts increase by appending $b_A(i)\circ b_B(i)$, so they remain identical.
Moreover, via Lemma~\ref{lem:nd1} we have that $b_A(i) = \pi(x\mid \calT^{A}(i))$ and that
$b_B(i)= \pi(y\mid \calT^{A}(i)\circ b_A(i))$. 
Hence, we have appended 
\[b_A(i)\circ b_B(i)=\pi \left(x\mid \calT^{A}(i)\right)\circ \pi \left(y\mid \calT^{A}(i)\circ\pi(x\mid \calT^{A}(i))\right)\]
 to both transcripts, which is exactly the correct continuation with respect to~$\calT^{\pi}$. Thus, both are still the correct prefixes of~$\calT^{\pi}$.

\item [Case 2: $r_B(i)=r_A(i)+1$.] In this case, no matter what Bob receives--silence, erasure or a message (with the same parity as $r_B(i)$)--he does not execute Lines~\ref{line:adap:incb}--\ref{line:adap:appb} and hence does not change $r_B,\calT^{B}$.  
As above there are two cases. If Alice does not increase her $r_A$, then by 
Lemma~\ref{lem:nd1} she doesn't change $\calT^{A}$ also and both transcripts remains the same,
which, by the induction hypothesis, satisfies Item~\ref{it:two}.

Now let's assume Alice increases~$r_A$, that is, $r_A(i+1)=r_A(i)+1$.
This can happen only if Bob sends a message $(b_B(i),r_B \bmod 2)$ with $r_B = r_A(i)+1 = r_B(i)$
(line~\ref{line:adap:aif2}) and Alice receives this correctly. 
Therefore, Alice sets $\calT^{A}(i+1) = \calT^A(i) \circ b_A(i) \circ b_B(i)$. 
We now claim that this is the correct continuation, i.e., that $\calT^{A}(i+1) = \calT^{B}(i)$. Clearly, $b_A(i) = \pi(x\mid \calT^A(i))$ which is the correct continuation of $\calT^{A}$ with respect to $\calT^{\pi}$.

The bit sent by Bob, $b_B(i)$, is also the correct continuation even if this message was generated at a previous round~$k<i$ and was re-transmitted at round~$i$. 
\begin{claim}If $r_A(i+1)=r_A(i)+1$ then, $b_B(i)=\pi(y\mid \calT^{A}(i)\circ\pi(x\mid \calT^{A}(i)))$. \end{claim}
\begin{proof}
Let $k\leq i-1$ be the largest integer such that $r_B(k)=r_A(k)$ but $r_B(k+1)=r_A(k+1)+1$ and thus, $r_
A(k+1)=r_A(k)$. 
By the induction hypothesis, such a round must exist since $r_A(1)=r_B(1)=0$ and $r_B(i)=r_A(i)+1$. 
Observe that there does not exist any round  $k'\in [k+1,i-1]$ satisfying $r_A(k')=r_B(k')$, since the existence of such~$k'$
contradicts the maximality of~$k$.  
Hence for all such $k'$, it holds that $r_B(k')=r_A(k')+1$. 
It follows that $r_B$ (and therefore $r_A$) cannot increase in any such round: as discussed in the beginning of \textbf{Case~2}, $r_B$ does not increase in this case; $r_A$ cannot increase for otherwise we will have $k'+1\in [k+2,i]$ for which $r_A(k'+1)=r_B(k'+1)$. In particular, 
$k'+1\in [k+1,i-1]$ since $r_B(i)=r_A(i)+1$. Recall that the existence of such $k'$ contradicts the maximality of $k$.

By the induction hypothesis, 
$\calT^{B}(k) = \calT^{A}(k)$.
Since $r_A$ does not change in rounds $[k,i-1]$, $\calT^{A}$ also does not change in those rounds (Lemma~\ref{lem:nd1}), and we get
$\calT^{B}(k) = \calT^{A}(k)=\calT^{A}(i)$.
Then,
\begin{align*}
\calT^{B}(k+1) &= \calT^{B}(k)\circ b_A(k)\circ b_B(k)\\
&=\calT^{A}(i)\circ b_A(k)\circ \pi\left(y\mid \calT^{A}(i)\circ b_A(k)\right)\\
&=\calT^{A}(i)\circ \pi\left(x\mid \calT^{A}(i)\right)\circ \pi\left(y\mid \calT^{A}(i)\circ\pi(x\mid \calT^{A}(i))\right),
\end{align*}
where the last equality follows because in round $k$, 
Alice must send a non-silent message with 
$b_{send}=b_{A}(k)=\pi(x\mid \calT^{A}(k)) = \pi(x\mid \calT^{A}(i))$  (line~\ref{line:adap:senda}).
Additionally, since Bob increases $r_B$ (line~\ref{line:adap:incb}) he also  sets $b_{send}=b_B(k)=\pi(y\mid \calT^{A}(i)\circ b_A(k))$ (line~\ref{line:adap:setb}), i.e.,
\[b_B(k)=\pi\left(y\mid \calT^{A}(i)\circ\pi(x\mid \calT^{A}(i))\right).\]
In rounds $[k+1,i]$ Bob does not increase $r_B$, hence he does not execute the if block (Lines~\ref{line:adap:incb}--\ref{line:adap:appb}) and does not changes his value of $b_{send}$. 
Hence, $b_{B}(i)=b_{B}(k)$, which completes the proof of the claim.
\end{proof}
\end{description}
\end{proof}

\begin{lemma}[Analog of Lemma~\ref{lem:end}]\label{lem:end1}
Any round~$i$, with $i > t_A$,  satisfies Item~\ref{it:one} of Lemma~\ref{lem:shiftsync1} and $r_A(i)=\frac{N}{2}$. Moreover, Bob remains silent in round $i$.
\end{lemma}
\begin{proof}
Alice terminates at $t_A$, hence at the end of this round she holds $r_A=N/2$ (line~\ref{line:adap:AliceExit}).
By definition $r_A$ remains the same for any round~$i$ after Alice has terminated. 
Since $r_A\le r_B$  (Lemma~\ref{lem:shiftsync1}) we also have that $r_B(i)=N/2$ for any such round~$i$. 
Note that it cannot increase above $N/2$ since Bob enters the termination phase once $r_B=N/2$ and never increases it again.

After Alice has terminated, Bob either hears silence or an erasure.
Therefore, for all rounds $i > t_A$, he never executes the \texttt{if} block of Line~\ref{line:adap:bif3}, since he never receives a message with the parity same as $r_B$. It follows that Bob executes the \texttt{else} block and remains silent indefinitely.
\end{proof}

Next, we show that the simulation eventually reaches to an end and does not get stuck or hangs indefinitely.
\begin{lemma}[Analog of Lemma~\ref{lem:progress}]\label{lem:progress1}
Let $i$ be any round for which Alice has not yet terminated (i.e., $r_A(i)<N/2$).
If no erasures at all occur in rounds $i,i+1$, then $r_A(i+2)\ge r_A(i)+1$.
\end{lemma}
\begin{proof}
Assume that Alice is not silent in round~$i$.
Then, consider the two cases of Lemma~\ref{lem:shiftsync1}. 
If $\calT^{A}(i)=\calT^{B}(i)$ and Alice sends a valid message $m_A=(b_A,r_A)$, this message arrives at Bob's end uncorrupted, and since it carries the correct parity, Bob replies with the correct bit and parity and 
Alice will increase her  $\calT^{A}$ and  $r_A$, so $r_A(i+1)=r_A(i)+1$.
If $\calT^{B}(i) = \calT^{A}(i)\circ b_A(i)\circ b_B(i)$, Alice will send a message with parity that
is incompatible with what Bob is expecting, and he will reply with a saved message that has exactly the parity Alice is expecting. Then, Alice extends her transcript and $r_A$, and we have $r_A(i+1)=r_A(i)+1$.

However, it is possible that Alice remains silent in round~$i$ (due to an erasure in the previous timestep). 
In this case, we again have two cases: either Bob repeats with a saved message (with parity $r_B(i)$); or Bob remains silent (if he switched to termination phase). 

\begin{description}
\item [If Bob is not in termination phase:]
Consider the two cases of Lemma~\ref{lem:shiftsync1}. 
If $r_B(i)=r_A(i)+1$ then this message is accepted by Alice: At the beginning of round $i$ 
Alice increases $r_A = r_A(i)+1$ and then the condition of line~\ref{line:adap:aif2} is satisfied.
Alice extends her transcript and does not decrease $r_A$ back (the \texttt{else} block), so $r_A(i+1)=r_A(i)+1$.
In the other case, $r_B(i)=r_A(i)$. Here, Alice remains silent which indicates Bob's message was erased.
Then, Bob replies with a message with parity $r_B=r_A(i)$, which is ignored by Alice since it has an incorrect parity. Then, in round $i+1$ Alice sends a message (i.e., she is not silence).
As argued above, if Alice is not silent and there are no erasures during that round, then $r_A$ increases.
Hence, $r_A(i+2)=r_A(i)+1$.
\item [If Bob is in termination phase:]
Alice being silent in round~$i$ only causes Bob to remain silent as well. Then, in round~$i+1$ Alice cannot remain silent, since her received message in round~$i$ was not~$\bot$ (Line~\ref{line:adap:aif1}). Then, the claim holds via the same argument as above when Alice is not silent, applied to round~$i+1$.
\end{description}
\end{proof}

The following is a corollary of the above proof.
\begin{corollary}\label{cor:progressA}
Assume a round $i$ in which Alice hasn't terminated yet and Bob hasn't switched to termination phase. 
If there are no erasures in round~$i$ then $r_A(i+1)=r_A(i)+1$ except for the case where 
$r_A(i)=r_B(i)$ and the second timestep of round~$i-1$ was erased.
\end{corollary}

Similar to the analysis in Section~\ref{sec:corr}, the above lemmas lead to the correctness of the scheme, which we complete in the proof of Theorem~\ref{thm:final1} below.
Before that, we turn to analyze the communication and round complexity of the scheme, which is a little more involved as compared to the Section~\ref{sec:corr}. To this end, we need a notion of (unit) costs incurred by Alice and Bob in order to send a message (recall that keeping silent comes with no cost), and (unit) cost incurred by the channel per erasure. We show in Lemma~\ref{lem:costbound} that the combined costs of Alice and Bob at any round does not exceed the length of Bob's partial transcript plus the cost incurred by the channel till that round.

In order to analyze the communication, we need the following definitions first. We define the (cumulative) cost incurred by Alice, till the beginning of round $i$ by $\ca(i)=|\{j\in [i-1]: m_A(j)\neq \square\}|$. Similarly, $\cb(i)=|\{j\in [i-1]: m_B(j)\neq \square\}|$. 
We also define the cost incurred by the channel, $\ch(i)$, as the total number of erasures that Alice and Bob both receive till the beginning or round $i$. 
\[
\ch(i) = |\{j \in [i-1]: m_{rec,A}(j)=\erasure  \} | + |\{j \in [i-1]: m_{rec,B}(j)=\erasure  \} |.
\]
We note that $\ca(t_A)+\cb(t_A)={\CC^{sym}(\Pi)}$ and $T\geq \ch(t_A)$. By definition, $\ca(i+1)\in\{\ca(i),\ca(i)+1\}$ and $\cb(i+1)\in\{\cb(i),\cb(i)+1\}$, i.e., each party speaks at most one symbol at every round.

The next lemma bounds the progress of the simulation as a function of the communication up to the observed noise.
\begin{lemma}\label{lem:costbound}
For any round~$i$ %
one of the following conditions holds:
\begin{enumerate}
\item\label{it:3}$r_B(i) = r_A(i)$ and $|\calT^{B}(i)| + \ch(i) \geq  \ca(i)+\cb(i)$ or,
\item\label{it:4}$r_B(i) = r_A(i)+1$ and $|\calT^{B}(i)| + \ch(i) \geq  \ca(i)+\cb(i)+1$.
\end{enumerate}
\end{lemma}
\begin{proof}
We will prove this lemma by induction on the round number.
In the first round, $|\calT^{B}(1)|$, $r_B(1),r_A(1)$, $\ch(1),\ca(1)$, and $\cb(1)$ are all zero which gives us Item~\ref{it:3}. 
Now we assume that the lemma holds for all rounds up to some $i\in [t_A]$ and prove that it must also hold at the end of round~$i$. Note that if the claim holds at the end of round~$t_A$, then it trivially holds afterwords, since Alice is not active anymore and Bob remains silent as given by Lemma~\ref{lem:end1}. 

For rounds where both parties are still running the scheme, we have the following case analysis.
\begin{description}
\item[Case 1: Item~\ref{it:3} holds in round $i$.] 
We start with the easiest case. For $i\le t_A$ we have $r_A(i) < N/2$, thus $r_B(i)<N/2$ and Bob is not yet in termination phase.
Suppose, $\cb(i+1)=\cb(i)$ which means Bob received an erasure in round~$i$ and kept silent. 
Bob does not change his~$r_B$, and thus Alice doesn't as well (otherwise $r_A$ will exceed~$r_B$).
It follows that $\calT^{B}(i+1)=\calT^{B}(i)$ while $\ca(i+1)\leq \ca(i)+1$ and $\ch(i+1)\geq \ch(i)+1$, hence Item~\ref{it:3} holds in round~$i+1$. 

Now consider that only Alice remains silent, that is, where $\cb(i+1)=\cb(i)+1$ but $\ca(i+1)=\ca(i)$. 
In this case $\calT^{B}$ does not change in round~$i$ since Bob either receives an erasure or silence, and thus $r_B$ and $r_A$ remain as is. 
On the surface of it, $\ca(i+1)+\cb(i+1)$ increased by one with respect to round~$i$ while
$|\calT^{B}(i+1)|+\ch(i+1)$ either increased by one (and then Item~\ref{it:3} holds), or did not change, which seems to be a problem. 
We show below that although such a case is possible, it can happen only if the inequality in Item~\ref{it:3} in round~$i$ was strict.
\begin{claim}
If $\ca(i+1)=\ca(i)$ and $\cb(i+1)=\cb(i)+1$ then   
$|\calT^{B}(i)| + \ch(i) \geq \ca(i)+\cb(i)+1$.
\end{claim}
\begin{proof}
Let $j\leq i-1$ be the largest round such that Alice chooses to speak in round~$j$.
Let $j'=j+1$. 
Since, for all $j' \leq k\leq i$, Alice remains silent in round $k$, she must have received an erasure in round $k-1$ and hence keeps $\calT^{A}$ unchanged in round $k-1$. Therefore, $\calT^{A}(j)=\calT^{A}(i)$ and by Lemma~\ref{lem:nd1}, $r_A(i)=r_A(j')=r_A(j)$. Moreover, for $k\geq j'$, Bob either receives an erasure in round $k$ or receives silence and consequently does not change $\calT^B$ whence $\calT^{B}(j')=\calT^{B}(i)$. 

By definition, Alice does not remain silent in round~$j$. Assume for the sake of contradiction that Bob receives Alice's message in round $j$ and changes $\calT^{B}$ based on this. Then he must increment $r_B$ and this implies $r_B(j')=r_B(j)+1$. By the above discussion and Lemma~\ref{lem:shiftsync1}, we have that 
\[
r_B(i)=r_B(j')=r_B(j)+1\geq r_A(j')+1=r_A(i)+1\text{,}
\] 
which gives us a contradiction because by assumption, $r_B(i)=r_A(i)$. Hence, Bob also receives an erasure in round $j$ which proves that $\ch(j')=\ch(j)+2$ and 
\begin{equation}\label{eq:case1:BobT}
\calT^{B}(j)=\calT^{B}(i). 
\end{equation}
As a result, Bob must remain silent in round $j$, $\cb(j)=\cb(j')$, and this silence must have been corrupted by the noise. Since Alice speaks at round $j$, $\ca(j')=\ca(j)+1$, with the above and the induction hypothesis we obtain
\begin{align*}
|\calT^{B}(j')| + \ch(j') &= |\calT^{B}(j)| + \ch(j)+2 \\
&\geq \ca(j)+\cb(j)+2 \\
&= \ca(j')+\cb(j')+1.
\numberthis
\label{eq:one}
\end{align*}
Since Alice keeps silent in rounds $[j',i]$ we know she must receive erasures in rounds $[j'-1,i-1]$. Recall that $\ch(t)$ includes all the erasures till the beginning of round $t$ and does not include the corruptions in round $t$. Thus,
 $\ch(i)\geq \ch(j')+i-j'$.  
 Then,
\begin{align*}
|\calT^{B}(i)| + \ch(i) &\geq |\calT^{B}(j')| + \ch(j')+ i-j' \\
&  \geq \ca(j')+\cb(j')+(i-j')+1, 
\numberthis\label{eq:two}
\end{align*}
where the first inequality is due to Eq.~\eqref{eq:case1:BobT} (and the monotony of the transcript's length), and the second inequality follows from Eq.~\eqref{eq:one}. 
Note that $\cb(i)\leq \cb(j')+i-j'$ as Bob may transmit at most one non-silence symbol from round $j'$ to round $i-1$; the transmission of round $i$ is not included in $\cb(i)$  
Together with Eq.~\eqref{eq:two}, we conclude that  $|\calT^{B}(i)| + \ch(i) \geq  \ca(i)+\cb(i)+1$.
\end{proof}
Going back to the proof of the lemma for the case of 
$\cb(i+1)=\cb(i)+1$ and $\ca(i+1)=\ca(i)$. As argued above, we have $r_A(i+1)=r_B(i+1)$ and the above claim gives that,  
\begin{align*}
|\calT^{B}(i+1)| + \ch(i+1) 
	&\geq |\calT^{B}(i)| + \ch(i) \\
	&\geq \ca(i)+\cb(i)+1 \\
	& =\ca(i+1)+\cb(i+1),
\end{align*}
which implies Item~\ref{it:3} holds for this case as well.

\goodbreak 
Lastly, if both $\ca(i)$ and $\cb(i)$ increase in round $i$ then Alice and Bob must have executed Lines~\ref{line:adap:senda} and~\ref{line:adap:sendb}, respectively. 
Therefore, Bob must have correctly received Alice's message and $|\calT^{B}(i+1)|=|\calT^{B}(i)|+2$ (Line~\ref{line:adap:appb}). If Alice received Bob's message correctly, it means there were no erasures in this round. So Alice extends her transcript as well, and increases $r_A(i)$ so now $r_A(i+1)=r_B(i+1)=r_A(i)+1$. Also, since there were no erasures, we have $\ch(i+1)=\ch(i)$. Putting them all together we have
\begin{align*}
|\calT^{B}(i+1)|+\ch(i+1) &= |\calT^{B}(i)|+\ch(i)+2 \\
&\ge \ca(i)+\cb(i)+2 \\
&= \ca(i+1)+\ca(i+1),
\end{align*}
where the inequality is the induction hypothesis. Hence, Item~\ref{it:3} holds in round $i+1$. 

However, if Alice received an erasure in round~$i$ she keeps $r_A(i)$ unchanged.
Then we have $\ch(i+1)=\ch(i)+1$ and $r_B(i+1) = r_A(i+1)+1$ while $r_A(i+1)=r_A(i)$. 
In this case Item~\ref{it:4} is satisfied:
\begin{align*}
|\calT^{B}(i+1)| + \ch(i+1)&=|\calT^{B}(i)|+2+\ch(i)+1 \\
&\geq \ca(i)+\cb(i) +3\\
&\geq \ca(i+1)+\cb(i+1)+1.
\end{align*}

\item[Case 2: Item~\ref{it:4} holds in round $i$.] 
First, observe that $r_B(i+1)=r_B(i)$ and so $\calT^{B}(i+1)=\calT^{B}(i)$. Again, we start with the easy case---when at most one of $\ca(i)$ and $\cb(i)$ increase in round $i$. If neither increase then Item~\ref{it:4} trivially holds in round $i+1$. 
Suppose, $\ca$ increases but $\cb$ does not then it must be the case that Bob received an erasure, or otherwise he would have replied with a message (even if Bob switched to termination phase).
Having an erasure means that $\ch(i+1)=\ch(i)+1$, so Item~\ref{it:4} continues to hold. Otherwise, $\cb$ increases but $\ca$ is unchanged. In this case, if $\ch(i)=\ch(i+1)$ then Alice receives Bob's message correctly and sets $r_A(i+1)=r_A(i)+1$. Hence, in round $i+1$,
\begin{align*}
|\calT^{B}(i+1)| + \ch(i+1)&=|\calT^{B}(i)|+\ch(i) \\
&\geq \ca(i)+\cb(i) +1\\
&\geq \ca(i+1)+\cb(i+1),
\end{align*} 
and Item~\ref{it:3} is satisfied.

We are now left with the case when both $\ca$ and $\cb$ increase. Here it is possible that the noise does not corrupt Alice's transmission in round $i$ but if this happens then the inequality in Item~\ref{it:4} in round~$i$ must be strict.
\begin{claim}
If $\ca(i+1)=\ca(i)+1$ and $\cb(i+1)=\cb(i)+1$ then 
$|\calT^{B}(i)| + \ch(i) \geq \ca(i)+\cb(i)+2$.
\end{claim}
\begin{proof}
The proof is very similar to that of the claim in case 1. Let $j\leq i-1$ be the largest round such that Bob chooses to speak in round $j$ and $j'=j+1$. 
Since, for all $j' \leq k\leq i$, Bob remains silent in round $k$, he must have received an erasure and hence keeps $\calT^{B}$ unchanged in round $k$. Therefore, $\calT^{B}(j')=\calT^{B}(i)$ and by Lemma~\ref{lem:nd1}, $r_B(i)=r_B(j')$. Moreover, for $k\geq j'$, Alice either receives an erasure in round $k$ or receives silence and consequently does not change $\calT^{A}$ which gives $\calT^{A}(j')=\calT^{A}(i)$.
By the above discussion we have 
\begin{align*}
r_B(j')=r_B(i)=r_A(i)+1=r_A(j')+1,
\end{align*}
hence,
\begin{align}
 |\calT^{B}(j')| + \ch(j') \geq  \ca(j')+\cb(j')+1. \label{eq:ind}
\end{align} 

Since Bob does not remain silent in round $j$, we conclude that Alice does not receive Bob's message correctly in round $j$. Since she receives an erasure, she sends silence in round~$j'$. However, Bob also remains silent in round $j'$ so Alice's silence must have been corrupted by the noise. That is, $\ch(j'+1)\geq \ch(j')+1$ but $\ca(j'+1)=\ca(j')$ and $\cb(j'+1)=\cb(j')$. By the above discussion and Eq.~\eqref{eq:ind},
\begin{align}
|\calT^{B}(j'+1)| + \ch(j'+1) &\geq |\calT^{B}(j')| + \ch(j')+1 \nonumber\\
&\geq \ca(j')+\cb(j')+2 \nonumber \\ 
&= \ca(j'+1)+\cb(j'+1)+2. \label{eq:three}
\end{align}
Since Bob keeps silent in rounds $[j',i-1]$ (also $[j'+1,i-1]$) we know he must receive erasures in rounds $[j',i-1]$ (also $[j'+1,i-1]$). We note again that $\ch(t)$ includes all the erasures till the beginning of round $t$ and does not include the corruptions in round $t$. Thus, $\ch(i)\geq \ch(j'+1)+i-(j'+1)$.
\begin{align}
|\calT^{B}(i)| + \ch(i) &\geq |\calT^{B}(j'+1)| + \ch(j'+1)+ i-(j'+1) \nonumber \\
&\geq \ca(j'+1)+\cb(j'+1)+(i-(j'+1))+2. \label{eq:four} 
\end{align}
We also know that $\ca(i)\leq \ca(j'+1)+i-(j'+1)$ (again $\ca(i)$ does not include the transmission in round $i$). Using Eq.~\eqref{eq:four}, we conclude that $|\calT^{B}(i)| + \ch(i) \geq  \ca(i)+\cb(i)+2$.
\end{proof}
Therefore, if $\cb(i+1)=\cb(i)+1$ and $\ca(i+1)=\ca(i)+1$, we have two cases. Either $\ch(i+1)=\ch(i)$ which implies $r_A(i+1)=r_B(i+1)$ and  then  Item~\ref{it:3} is satisfied, since 
\begin{align*}
|\calT^{B}(i+1)| + \ch(i+1) &\geq |\calT^{B}(i)| + \ch(i) \\
&\geq \ca(i)+\cb(i)+2 \\
&=\ca(i+1)+\cb(i+1).
\end{align*}
Otherwise, $\ch(i+1)=\ch(i)+1$ and $r_B(i+1)=r_A(i+1)+1$ and Item~\ref{it:4} is satisfied,
\begin{align*}
|\calT^{B}(i+1)| + \ch(i+1)&=|\calT^{B}(i)|+\ch(i)+1\\
&\geq \ca(i)+\cb(i) +3\\
&= \ca(i+1)+\cb(i+1)+1.
\end{align*} 
\end{description}
\end{proof}

We are ready to complete the main proof of this section, and show that the scheme correctly simulates any~$\pi$ with a low amount of communication. 
\begin{theorem}\label{thm:final1}
Let $\pi$ be an alternating binary protocol. 
There exists a coding scheme $\Pi$ with semi-termination over a $4$-ary alphabet such that for any instance of~$\Pi$ that suffers at most $T$ erasures overall (for $T\in \mathbb{N}$ an arbitrary integer), Alice and Bob output $\calT^{\pi}$. Moreover, $\CC^{sym}(\Pi)\leq\CC(\pi)+T$ and $\RC(\Pi)\le\RC(\pi)+4T$.
\end{theorem}
\begin{proof}
Lemma~\ref{lem:shiftsync1} guarantees that at every given round, Alice and Bob hold a correct prefix of~$\calT^{\pi}$. Moreover, we know that by the time Alice terminates, her transcript (and hence Bob's transcript) is of length at least~$N$, which follows from Lemmas~\ref{lem:end1} and~\ref{lem:nd1}, i.e., from the fact that $r_A=N/2$ at the termination, and that every time~$r_A$ increases by one, the length of~$\calT^{A}$ increases by two. Finally, note that if the number of erasures is bounded by~$T$, then Alice will eventually reach termination. This follows since after $T$~erasures have happened, $r_A$ increases by one in every round (Lemma~\ref{lem:progress1}, Corollary~\ref{cor:progressA}),
maybe up to the last round where Bob switches to termination phase, which may take another round of communication to increase~$r_A$. If so, $r_A$ eventually reaches~$N/2$ and Alice terminates. 

Regarding the communication complexity, note that
when Alice exits, 
we have that $r_A(t_A)=r_B(t_A)=N/2$ as well as 
$|\calT^{A}(t_A)|=|\calT^{B}(t_A)|=N$; recall that $N=\CC(\pi)=\RC(\pi)$ is the length of the noiseless protocol~$\pi$ that we want to simulate; this follows from the above correctness argument. 
From Item~\ref{it:3} of Lemma~\ref{lem:costbound}, we know that $|\calT^{B}(t_A)| + \ch(t_A) \geq  \ca(t_A)+\cb(t_A)$. We also know from Lemma~\ref{lem:end1} that Bob remains silent after Alice exits,
thus $\CC^{sym}(\Pi)=\ca(t_A)+\cb(t_A)$. Therefore, 
\[
\CC^{sym}(\Pi)=\ca(t_A)+\cb(t_A) \leq  |\calT^{B}(t_A)| + \ch(t_A)  \leq \CC(\pi)+T.
\]

Finally, let us analyze the round complexity, $\RC(\Pi)$, defined as the number of timesteps until  Alice terminates 
(recall that Bob never terminates).
Lemmas~\ref{lem:RC} and \ref{lem:RCterm} below suggest that
the worst case (with respect to the round-complexity)
is when all the erasures occur after Bob has switched to termination phase,
while Alice hasn't yet simulated the last bit of the protocol.
In this case we have that $r_A$ reaches $N/2$ after at most
$N/2 + 2T$ rounds, and since every round has two timesteps, we have
\(
\RC(\Pi) = 2t_A \le N+4T.
\)
\end{proof}

\begin{lemma}\label{lem:RC}
If Bob hasn't switched to termination phase, at round~$i$, $r_A(i) \ge (i-1)-\ch(i)$. 
\end{lemma}
\begin{proof}
Consider $r_A(i)$ for round~$i\in [t_A]$. 
If there was an erasure, then we know that $r_A(i+1)=r_A(i)$. 
If there were no erasures at round~$i$, then by Corollary~\ref{cor:progressA} we know that $r_A(i+1)=r_A(i)+1$ except for a single case we discuss below.

The only case where no erasures happen yet $r_A(i)$ does not increase is when the second timestep of round $i-1$ is erased conditioned on having $r_A(i-1)=r_B(i-1)$. Notice that since Bob's message at round $i-1$ is erased, we must have that $r_A(i)=r_B(i)$. 
Indeed, only Bob can increase his $r_B$ in round $i-1$ (Alice receives an erasure and cannot increase~$r_A$). Suppose $r_A(i-1)+1=r_B(i-1)$, then we cannot have $r_A(i)=r_B(i)$ unless Alice increases her $r_A$, which is a contradiction.

Next, note that $r_A(i-1)=r_B(i-1)$ and $r_A(i)=r_B(i)$ can happen only if the \emph
{first} timestep of round $i-1$ was either erased or a silence was transmitted---otherwise, Bob would have increased $r_B$ (but Alice wouldn't have) and so $r_B(i)=r_A(i)+1$. 
If the first timestep was erased, then in round $i-1$ there were two erasures rather than one.
If the first timestep was a silence, then the second timestep of round $i-2$ must have been erased, and we can apply the same argument on round $i-2$ inductively, until we reach a round where both timesteps were erased.

We conclude that if there are no erasures then $r_A$ increases every round, hence at the beginning 
of round~$i$, $r_A(i)$ would be at least~$i-1$. 
A single erasure halts the increase of~$r_A$ by a single round, unless there is a chain
of erased rounds and that chain begins with 2 erasures.
In this case $r_A$ does not increase throughout the chain as well as during the round following the chain. The number of erasures in the chain is also at least the length of the chain plus one. 
\end{proof}
\begin{lemma}\label{lem:RCterm}
If Bob has switched to termination phase each erasure causes at most two rounds in which $r_A$ does not increase.
\end{lemma}
\begin{proof}
The worst case is the following. Alice sends a valid message and this message arrives at Bob's side correctly. Now Bob replies with the last bit of the simulation (line~\ref{line:adap:bobLast}). Assume this message is being erased. Alice, seeing an erasure, remains silent (the \texttt{if} of line~\ref{line:adap:aif1} goes into the \texttt{else} block). When Bob sees a silence in termination phase, he remains silent as well (line~\ref{line:adap:bif3} goes into the \texttt{else} block). Hence, two rounds have passed without Alice increasing~$r_A$, due to a single erasure. In the next round she will again try to send a message.

It is easy to verify all other cases cause at most the above delay of two rounds per erasure.
\end{proof}

\subsection{Reducing the alphabet size} \label{sec:unary}
The above scheme uses an alphabet of size 4 (in addition to silence), which respectively increases the communication complexity (measured in bits). We now show how to reduce the alphabet size so it is unary, that is, the parties either send ``energy'' (1) or remain silent~(0). The communication complexity for this case is defined to be the \emph{energy complexity}---the number of rounds in which energy was transmitted.

Towards this goal, the parties carry out a certain type 
of temporal encoding~\cite{AGS16} described below.
We argue that Algorithms~\ref{alg:sima1} and~\ref{alg:simb1} concatenated with the temporal encoding satisfy Theorem~\ref{thm:main2}. 
\begin{theorem}\label{thm:final2}
Let $\pi$ be an alternating binary protocol. 
There exists an (AGS) adaptive coding scheme $\Pi_1$ with semi-termination using a unary alphabet such that for any instance of~$\Pi_1$ that suffers at most $T$ erasures overall (for $T\in \mathbb{N}$ an arbitrary integer), Alice and Bob output~$\calT^{\pi}$.

Moreover, it holds that  $\CC(\Pi_1)\leq \CC(\pi)+T$ and $\RC(\Pi_1)  \leq 4(\RC(\pi)+4T)$.
\end{theorem}
\begin{proof}
In the coding scheme $\Pi_1$, Alice and Bob will simulate each timestep of Algorithms~\ref{alg:sima1} and~\ref{alg:simb1} (coding scheme $\Pi$) as a  ``block'' that contains four timesteps that all belong to the same party. 
When a party wishes to send a message $(b,\rho)\in\{0,1\}\times\{0,1\}$ it simply transmits a ``1'' in the $(b+2\rho+1)$-th timestep of that block and remains silent in the other timesteps that belong to the same block. 
If a party wishes to remain silent, it just keeps silent throughout the entire block.
Each block of the run of $\Pi_1$ is decoded to a timestep of $\Pi$ in a manner such that if the block contains a single $1$, then $(b,\rho)$ can be recovered according to its position within the block. If all timesteps are silent, then the block decodes to~$\silence$. In all other cases, the block decodes to~$\erasure$.

Any instance of $\Pi_1$---specified by a given erasure pattern and inputs---can be directly mapped to an instance of $\Pi$ using the above encoding.
Moreover, the resulting instance of $\Pi$ has at most $T$~erasures. This is because if any block in $\Pi_1$ has more than one erasure, the corresponding timestep in $\Pi$ has a single erasure only. From Theorem~\ref{thm:final1}, we know that in the resulting instance of $\Pi$ having at most of $T$ erasures, Alice and Bob both compute the correct output~$\calT^{\pi}$. 
Therefore, we can conclude that in the given instance of $\Pi_1$ both players compute the correct output~$\calT^{\pi}$. 

From the above mapping, we see that each symbol of~$\Pi$ translates to a single (non-silent) transmission in~$\Pi_1$. Then we obtain $\CC(\Pi_1)=\CC^{sym}(\Pi)\leq \CC(\pi)+T$. Lastly, each timestep in $\Pi$ is mapped to four timesteps in $\Pi_1$, hence, $\RC(\Pi_1) = 4\RC(\Pi) \leq 4(\RC(\pi)+4T)$. 
\end{proof}

%% file: appendix.tex
\section{Unsynchronized Termination}\label{app:modeldetails}
In this section we discuss termination and the effect noise has on termination.
We argue that, in the presence of noise, no protocol can terminate in a ``coordinated way'', that  is, with both parties terminating in exactly same round. This is an artifact of
attaining common knowledge, and is closely related to the ``coordinated attack problem''~\cite{HM90}.

In particular, we show that an unbounded noise can always lead to a situation where one party terminates
while the other does not.  
Noise that occurs after this point where only a single party has terminated,
can still cause damage to the other party (i.e., prevent it from terminating).
Hence, such noise must be counted towards the adversary's noise budget. 
This intuition is formalized in the following lemma.

\begin{lemma}\label{lem:NoiseAfterTerm}
If  corruptions made after one party terminates are not counted towards the adversary's budget, then there exists a function $f$ such that 
for any protocol $\Pi$ that computes~$f$ and is resilient to any amount of noise $T\in \N$, at least one of the parties never terminates.
\end{lemma}

In order to prove the lemma, we
consider the simple \emph{bit-exchange} task, where Alice and Bob hold a single bit each, 
and they wish learn the bit of the other party. Assuming a noiseless channel, this can be done 
via two transmissions of one bit each. 
We try to construct a protocol for the exchange-bit task which is resilient to erasures.
On top of exchanging the input bits, we require a special property from our protocol, which we call
\emph{coordinated termination}---%
we require that in any instance of the protocol the parties terminate in the same round regardless of the observed noise. 

\newcommand{\cterm}{coordinated-terminaiton\xspace}
\begin{definition}
A coordinated protocol is one in which Alice and Bob \emph{always} terminate at the same round.
\end{definition}
\begin{definition}
We say that Alice and Bob have \cterm at round~$i$ if there exist $x,y$ and a finite noise pattern~$P$
for which both Alice and Bob terminate in round~$i$ given that their input is~$(x,y)$ and the noise is~$P$.
\end{definition}

The protocols we consider must be resilient to an unbounded (yet, finite) amount of noise.
\begin{definition}
A protocol is \emph{resilient to an unbounded amount of noise} if
for all $(x,y)$ and any finite noise pattern~$P$, the parties terminate and output the correct
output given they have the inputs~$(x,y)$ and the observed noise is~$P$. 
\end{definition}
We note that any function~$f$ can be computed in~$N$ rounds, (with $N=N(f)$), as long as at most $N/2$ of the transmissions are erased~\cite{FGOS15,EGH16}. Such protocols are coordinated since both parties terminate at round~$N$. 
However, in our setting where unbounded number of erasures may occur, 
the parties cannot predetermine their termination time to~$N$
since the number of corruptions $T$ may exceed~$N$. 
In a hindsight, the fact that termination time cannot be predetermined and depends on the observed noise,
in addition to the fact that parties observe different noise pattern, implies that no coordinated termination
is possible.

Our main lemma argues that no protocol for the bit-exchange task is both coordinated and resilient to unbounded number of erasures.
\begin{lemma}\label{lem:termination}
Let $\Pi$ be a coordinated protocol for bit-exchange, which is resilient to an unbounded amount~$T$ of erasures.
There always exists a (finite) noise pattern for which Alice and Bob don't terminate in a coordinated way.
\end{lemma}

We prove the lemma via the following claims.
\begin{claim}\label{clm:term-base}
Any coordinated protocol for the bit-exchange task which is resilient to an unbounded amount of erasures can never have a \cterm in rounds $i=1$ or $i=2$.
\end{claim}
\begin{proof}
For $i=1$ it is trivial that the parties cannot be correct if they terminate at $i=1$ since Alice never learns Bob's bit. For $i=2$, suppose that we have
a protocol~$\Pi$ that satisfies the statement, where the  \cterm is witnessed by some $x,y$ and $P=P_1P_2$; that is, both parties terminate at $i=2$ with the output $x,y$ when the noise is~$P$. 
We argue that $P=00$ or otherwise the protocol cannot be correct. 
In particular, if $P_1\ne 0$ Bob never gets any information from Alice and therefore he must give a wrong output given an instance where Alice holds any other $x'\ne x$ with the same noise $P$. The case for $P_2\ne0$ is similar; Hence $P=00$.

To ease notation, we write $\Pi(x,y \| P)$ to denote an instance of the protocol on inputs~$(x,y)$ with observed noise pattern~$P$. Recall that the parties speak alternatingly and denote by $m_1$ the message (that Alice sends) in round $i=1$, and by $m_2$ the message (that Bob sends) in $i=2$, etc.
The \emph{view} of a party up to some round~$i$, includes its input as well as all the messages received by round~$i$. Note that the behaviour of a party at round~$i$ is a deterministic function of its view at the beginning of that round. 

Now, we claim that Bob cannot terminate at $i=2$ because at that point in  the protocol he doesn't know whether Alice has received his message~$m_2$, or not. 
Suppose towards contradiction that Bob terminates in round $i=2$ given the above instance;
due to the coordinated termination, Alice must also terminate at $i=2$ given the view $(x, m_2=\bot)$.

Now, consider a run on the inputs $(x,y')$ with noise $P'=01$. 
If Bob still terminates in round $i=2$ on $\Pi(x,y'\| P')$ 
then Alice must be wrong on
either $x,y$ or $x,y'$ since she never receives any information from Bob.
If Bob doesn't terminate in $i=2$ then $\Pi$ is not coordinated (and does not satisfy the conditions in the Claim's statement). This is because,  Alice \emph{does} terminate in
$i=2$, since her view in $\Pi(x,y'\| P')$ is $(x,m_2=\bot)$. 
This view is identical to her view in $\Pi(x,y\| P')$ and as mentioned above, given this view Alice terminates in round $i=2$.
\end{proof}

\begin{claim}\label{clm:term-induc}
Suppose 
that any coordinated and resilient to an unbounded number of erasures
protocol~$\Pi$ 
cannot have \cterm in rounds $1,\ldots,i$. Then, any such protocol cannot have a \cterm in round $i+1$.
\end{claim}
\begin{proof}
Assume towards contradiction that a coordinated protocol~$\Pi$ that is resilient to an unbounded number of erasures does have \cterm in round $i+1$, while it is guaranteed that any other such protocol can never have a coordinated termination in rounds~$1,\ldots,i$.
Since $\Pi$ has a \cterm in $i+1$ we know that 
there exist $(x,y)$ and $P$ such that in $\Pi(x,y\| P)$, both parties terminate at $i+1$ with the correct output.

Without loss of generality assume $m_{i+1}$ is sent by Alice.
Consider $P'= P \vee 0^{i}1$. 
We examine  the two different cases of whether Bob halts in round $i+1$, or not.
\begin{description}
\item[1. Bob does not halt in $i+1$.]
In this case  $\Pi$ is not coordinated:  Alice will terminate in round $i+1$ since her view 
in $\Pi(x,y \| P')$ is identical to her view in~$\Pi(x,y\|P)$ in which she terminates in round~$i+1$ by  assumption.
\item[2. Bob halts in $i+1$.]
In this case 
we claim that there exists a coordinated and resilient protocol $\Pi'$ that has \cterm in round~$i$, in contradiction to the claim's assumption. 

Let 
\(
E=\{ (u,v,T) \mid \Pi(u,v\| T) \text{ terminates in~$i+1$} \}
\) 
be the set of all the triples $(u,v, T)$ such that $\Pi$ terminates in~$i+1$ (note that both parties terminate in round $i+1$ due to $\Pi$ being coordinated).
Define a new protocol, $\Pi'$ that behaves exactly like $\Pi'$, except for the case where a party has a view that is identical to $\Pi(u,v\|T)$ with $(u,v,T)\in E$ up to round $i$ (for Alice) or up to round $i-1$ (for Bob). 
In such cases, the parties in $\Pi'$ halt in round~$i$;  Alice outputs whatever she outputs in~$\Pi$, and Bob outputs the same value he outputs in~$\Pi$ given that $m_{i+1}=\bot$.

We claim that $\Pi'$ computes the bit-exchange task, and furthermore, 
$\Pi'$ is coordinated and resilient to an unbounded number of erasures.
First, let us argue $\Pi'$ is coordinated if $\Pi$ is. Assume towards contradiction 
that there is an instance $\Pi'(u,v \| T)$ for which the parties don't have a \cterm.
Since $\Pi'$ behaves like $\Pi$ (and $\Pi$ is normally coordinated), 
this is possible only if one of the parties (but not both!), has a view that is identical
to a view of some instance $\Pi(u',v'\|T')$ with $(u',v',T')\in E$. 

Assume (without loss of generality) that this party is Alice.
Note that if Alice terminates in round~$i$ in~$\Pi'(u,v \| T)$ 
then she terminates in round~$i+1$ in $\Pi(u,v \|T)$ by the definition of $\Pi'$.
Since $\Pi$ is coordinated, Bob terminates in round $i+1$ in $\Pi(u,v\|T)$ as well. 
It follows that $(u,v\|T)\in E$. Hence, Bob also terminates in round~$i$ in $\Pi'(u,v \| T)$
by definition, and we reached a contradiction.

Now let us argue that $\Pi$ is correct and computes the same output as~$\Pi$ for any inputs and noise.
To see this, assume toward contradiction an instance
$\Pi'(u,v \| T)$ for which one of the parties is incorrect, i.e., it outputs a different value than in~$\Pi(u,v\|T)$. 
Since $\Pi'$ usually performs exactly like $\Pi'$,
it must be the case where either Alice or Bob has terminated at round $i$ in $\Pi'(u,v \|T)$ while they don't do so in $\Pi(u,v \| T)$, i.e., $(u,v,T)\in E$. 
Assume Alice terminates in round~$i$. 
This means that in $\Pi'(u,v \| T)$ Alice terminates in round $i+1$, and since she is the sender of $m_{i+1}$, she must have the correct output by round~$i$.
On the other hand, if Alice terminates in round~$i$ without sending~$m_{i+1}$, maybe Bob is missing the information of this last message in order to give the correct output. This can't be the case, by considering the instance $\Pi(u,v \| T \vee 0^i1)$ in which Bob never receives $m_{i+1}$ but still terminates in round $i+1$.
Indeed, 
Alice terminates in~$i+1$ in $\Pi(u,v \| T)$ as argued above; from Alice's point of view, $\Pi(u,v \| T)$ is identical to her view in~$\Pi(u,v \| T \vee 0^i1)$, thus she must terminate in round~$i+1$ in~$\Pi(u,v \| T \vee 0^i1)$ as well. Since $\Pi$ is coordinated,
Bob terminates in round~$i+1$ in~$\Pi(u,v \| T \vee 0^i1)$ as argued. We know that Bob's view till round $i+1$ in $\Pi(u,v \| T \vee 0^i1)$ is the same as Bob's view till round $i-1$ in $\Pi'(u,v \| T)$. Since Bob correctly terminates in in $\Pi(u,v \| T \vee 0^i1)$, he must also correctly terminate in $\Pi'(u,v \| T)$.

The analysis of the other case, where Bob terminates in~$i$, is similar (and in any case we have proved that $\Pi'$ is coordinated).

Finally, note that $(x,y,P)\in E$, i.e., $E$ is non-empty.
Since for any $(u,v,T) \in E$, $\Pi'$ terminates in round~$i$,
we have that $\Pi'$ has a \cterm in round~$i$ witnessed by $(x,y,P)$ contradicting the assumption.
\end{description}
\end{proof}

We can now complete the proof of Lemma~\ref{lem:termination}.
\begin{proof}[Proof of Lemma~\ref{lem:termination}.]
The lemma is an immediate corollary of the above claims. 
A coordinated and resilient protocol (for the bit-exchange task) cannot terminate in rounds $i=1$ or $i=2$ as stated by
Claim~\ref{clm:term-base}.
Furthermore, Claim~\ref{clm:term-induc} serves as the induction step and proves that 
no such protocol
can have a \cterm in any round~$i>2$. Hence, if we have a resilient protocol it cannot be coordinated. 
Namely, there will be a certain situation (inputs and noise) in which the parties terminate in different rounds.
\end{proof}

Finally, we can complete the proof of Lemma~\ref{lem:NoiseAfterTerm}.
\begin{proof}[Proof of Lemma~\ref{lem:NoiseAfterTerm}]
Assume a protocol in which each party terminates at the first round where they are certain that both parties have computed the correct output.
Lemma~\ref{lem:termination} prove that there is always a noise pattern that causes one party (without loss of generality, Alice) to terminate while the other party does not terminate. At the termination time of Alice, Bob is still uncertain that both sides have learned the output (otherwise, he would have terminated as well).
Now, assume that Bob keeps receiving erasures---his information remains the same and thus his uncertainty  remains.
It follows that Bob could not terminate as long as it sees erasures.
Hence, if erasures that arrive after Alice has terminated are not counted towards the noise budget, then the channel may produce erasures indefinitely, and prevent Bob from terminating indefinitely.
\end{proof}